\documentclass[a4paper,USenglish,cleveref,autoref,thm-restate]{lipics-v2021}

\nolinenumbers
\newtheorem*{imptremark*}{Important remark}

\title{A generic characterization of generalized unary temporal logic and two-variable first-order logic}

\titlerunning{A generic characterization of generalized unary \tls and two-variable \fo}

\author{Thomas Place}{Univ. Bordeaux, CNRS, Bordeaux INP, LaBRI, UMR 5800, F-33400, Talence, France\and \url{http://www.labri.fr/perso/tplace}}{tplace@labri.fr}{https://orcid.org/0009-0000-2840-9586}{}

\author{Marc Zeitoun}{Univ. Bordeaux, CNRS, Bordeaux INP, LaBRI, UMR 5800, F-33400, Talence, France \and \url{http://www.labri.fr/perso/zeitoun}}{mz@labri.fr}{https://orcid.org/0000-0003-4101-8437}{}

\authorrunning{T. Place and M. Zeitoun}

\Copyright{Thomas Place and Marc Zeitoun}



\category{}

\relatedversion{} \relatedversiondetails[cite=pz_gen_tl_csl]{Extended version}{}

\EventShortTitle{}
\EventYear{2024}
\EventLogo{}
\ArticleNo{6}

\usepackage{xspace}
\usepackage{bbm}
\usepackage{thmtools,thm-restate}
\usepackage{csquotes}
\usepackage{nowidow}
\usepackage{widows-and-orphans}
\usepackage{tikz}
\usetikzlibrary{decorations.pathreplacing,hobby,backgrounds,trees,arrows,automata,shapes,positioning,fit,patterns,calc,matrix}

\newcommand{\veps}{\ensuremath{\varepsilon}\xspace}
\newcommand{\inv}{^{-1}}
\newcommand{\nat}{\ensuremath{\mathbb{N}}\xspace}

\newcommand{\cont}[1]{\ensuremath{\mathit{alph}(#1)}\xspace}

\newcommand{\Cs}{\ensuremath{\mathcal{C}}\xspace}
\newcommand{\Ds}{\ensuremath{\mathcal{D}}\xspace}

\newcommand{\Gs}{\ensuremath{\mathcal{G}}\xspace}

\newcommand{\Ls}{\ensuremath{\mathcal{L}}\xspace}

\newcommand{\Rs}{\ensuremath{\mathcal{R}}\xspace}

\newcommand{\Xs}{\ensuremath{\mathcal{X}}\xspace}

\newcommand{\vari}{prevariety\xspace}

\newcommand{\varis}{prevarieties\xspace}

\newcommand{\infsig}[1]{\ensuremath{\mathbb{I}_{#1}}\xspace}
\newcommand{\infsigc}{\infsig{\Cs}}

\usepackage{tikz}
\usetikzlibrary{automata}

\tikzset{every state/.style={draw=blue!50!green,very thick,fill=blue!50!green!20}}
\tikzset{statesub/.style={state,minimum size=1.3cm,inner sep=1pt}}
\tikzset{pattstate/.style={state,draw=red!50!yellow,line width=2pt,fill=red!50!yellow!20}}
\tikzset{pdotstate/.style={state,minimum size=0.75cm,inner sep=0.5pt,draw=red!50!yellow,line
		width=2pt,dashed,fill=red!50!yellow!20}}
\tikzset{lstate/.style={state,minimum size=0.65cm,inner sep=0.5pt}}
\tikzstyle{trans}=[shorten >= 1pt,thick,->]

\makeatletter
\tikzstyle{initial by arrow}=   [after node path=
{
	{
		[to path=
		{
			[->,double=none,shorten >= 1pt,thick,every initial by arrow]
			([shift=(\tikz@initial@angle:\tikz@initial@distance)]\tikztostart.\tikz@initial@angle)
			node [shape=coordinate,anchor=\tikz@initial@anchor,draw] {\tikz@initial@text}
			-- (\tikztostart)}]
		edge ()
	}
}]
\makeatother

\makeatletter
\tikzstyle{accepting by arrow}=   [after node path=
{
	{
		[to path=
		{
			[->,double=none,shorten >= 1pt,thick,every accepting by arrow]
			(\tikztostart) --
			([shift=(\tikz@accepting@angle:\tikz@accepting@distance)]\tikztostart.\tikz@accepting@angle)
			node [shape=coordinate,anchor=\tikz@accepting@anchor,draw] {\tikz@accepting@text}
		}
		]
		edge ()
	}
}]
\makeatother

\let\omin\min

\renewcommand{\min}{\ensuremath{\text{\scriptsize min}}\xspace}

\newcommand{\mods}{\ensuremath{\mathit{MOD}}\xspace}

\newcommand{\fo}{\ensuremath{\textup{FO}}\xspace}
\newcommand{\fow}{\mbox{\ensuremath{\fo(<)}}\xspace}

\newcommand{\foc}{\ensuremath{\fo(\infsigc)}\xspace}

\newcommand{\fod}{\ensuremath{\fo^2}\xspace}

\newcommand{\fodc}{\ensuremath{\fod(\infsigc)}\xspace}

\newcommand{\fodw}{\ensuremath{\fod(<)}\xspace}
\newcommand{\fodws}{\ensuremath{\fod(<,+1)}\xspace}

\newcommand{\fodwm}{\ensuremath{\fod(<,\mods)}\xspace}

\renewcommand{\ul}{\ensuremath{\textup{UL}}\xspace}

\newcommand{\at}{\ensuremath{\textup{AT}}\xspace}

\newcommand{\sfr}{\ensuremath{\textup{SF}}\xspace}

\newcommand{\stzer}{\textup{ST}\xspace}
\newcommand{\dotzer}{\textup{DD}\xspace}

\newcommand{\cocl}[1]{\ensuremath{\mathit{co\textup{-}}\!#1}\xspace}
\newcommand{\bool}[1]{\ensuremath{\textup{Bool}(#1)}\xspace}
\newcommand{\pol}[1]{\ensuremath{\textup{Pol}(#1)}\xspace}
\newcommand{\bpol}[1]{\ensuremath{\textup{BPol}(#1)}\xspace}

\newcommand{\upol}[1]{\ensuremath{\textup{UPol}(#1)}\xspace}

\newcommand{\copol}[1]{\ensuremath{\cocl{\textup{Pol}(#1)}}\xspace}

\newcommand{\sfp}[1]{\ensuremath{\textup{\textup{SF}}(#1)}\xspace}

\newcommand{\upolo}{\ensuremath{\textup{UPol}}\xspace}
\newcommand{\bpolo}{\ensuremath{\textup{BPol}}\xspace}

\newcommand{\opo}{\ensuremath{\textup{Op}}\xspace}

\newcommand{\poln}{\ensuremath{\textup{Pol}}\xspace}
\newcommand{\bpoln}{\ensuremath{\textup{BPol}}\xspace}

\newcommand{\ldet}[1]{\ensuremath{\textup{LPol}(#1)}\xspace}
\newcommand{\rdet}[1]{\ensuremath{\textup{RPol}(#1)}\xspace}

\newcommand{\ldeto}{\ensuremath{\textup{LPol}}\xspace}
\newcommand{\rdeto}{\ensuremath{\textup{RPol}}\xspace}

\newcommand{\meda}{\ensuremath{\textup{\sffamily M}_e\textup{\sffamily DA}}\xspace}
\newcommand{\davar}{\ensuremath{\textup{\sffamily DA}}\xspace}

\newcommand{\tls}{\ensuremath{\textup{TL}}\xspace}
\newcommand{\tlxs}{\ensuremath{\textup{TLX}}\xspace}
\newcommand{\tla}[1]{\ensuremath{\tls[#1]}\xspace}

\newcommand{\tlc}[1]{\ensuremath{\tls(#1)}\xspace}

\newcommand{\tlsn}[1]{\ensuremath{\textup{TL}_{#1}}\xspace}

\newcommand{\tlnc}[2]{\ensuremath{\tlsn{#1}(#2)}\xspace}

\newcommand{\tlfo}{\ensuremath{\textup{FL}}\xspace}
\newcommand{\tlpo}{\ensuremath{\textup{PL}}\xspace}

\newcommand{\tlfoc}[1]{\ensuremath{\tlfo(#1)}\xspace}
\newcommand{\tlpoc}[1]{\ensuremath{\tlpo(#1)}\xspace}

\newcommand{\tlfoa}[1]{\ensuremath{\tlfo[#1]}\xspace}
\newcommand{\tlpoa}[1]{\ensuremath{\tlpo[#1]}\xspace}

\newcommand{\ltl}{\ensuremath{\textup{LTL}}\xspace}

\newcommand{\ltlc}[1]{\ensuremath{\ltl(#1)}\xspace}

\newcommand{\finally}[1]{\mbox{\ensuremath{\textup{F}\,#1}}\xspace}
\newcommand{\nex}[1]{{\ensuremath{\textup{X}\,#1}}\xspace}

\newcommand{\finallyp}[2]{\mbox{\ensuremath{\textup{F}_{#1}\,#2}}\xspace}
\newcommand{\finallymp}[2]{\ensuremath{\textup{P}_{#1}\,#2}\xspace}
\newcommand{\finallyl}[1]{\finallyp{L}{#1}}
\newcommand{\finallyml}[1]{\finallymp{L}{#1}}

\newcommand{\xfinallyp}[2]{\ensuremath{\textup{F}^{sa}_{#1}\,#2}\xspace}
\newcommand{\xfinallymp}[2]{\ensuremath{\textup{P}^{sa}_{#1}\,#2}\xspace}

\newcommand{\tleqp}[1]{\ensuremath{\cong_{\eta,#1}}\xspace}
\newcommand{\tleqk}{\tleqp{k}}

\newcommand{\tlfeqp}[1]{\ensuremath{\blacktriangleright_{\eta,#1}}\xspace}
\newcommand{\tlfeqk}{\tlfeqp{k}}
\newcommand{\tlfdqp}[1]{\ensuremath{\blacktriangleright_{\delta,#1}}\xspace}
\newcommand{\tlfdqk}{\tlfdqp{k}}

\newcommand{\Jrel}{\ensuremath{\mathrel{\mathscr{J}}}\xspace}

\newcommand{\Rrel}{\ensuremath{\mathrel{\mathscr{R}}}\xspace}
\newcommand{\Lrel}{\ensuremath{\mathrel{\mathscr{L}}}\xspace}

\newcommand{\Jord}{\ensuremath{\leqslant_{\mathscr{J}}}\xspace}

\newcommand{\Rord}{\ensuremath{\leqslant_{\mathscr{R}}}\xspace}
\newcommand{\Lord}{\ensuremath{\leqslant_{\mathscr{L}}}\xspace}

\newcommand{\Jords}{\ensuremath{<_{\mathscr{J}}}\xspace}

\newcommand{\Rords}{\ensuremath{<_{\mathscr{R}}}\xspace}
\newcommand{\Lords}{\ensuremath{<_{\mathscr{L}}}\xspace}

\newcommand{\poschar}{\textup{\sffamily{Pos}}}

\newcommand{\pos}[1]{\ensuremath{\poschar(#1)\xspace}}
\newcommand{\posc}[1]{\ensuremath{\poschar_{c}(#1)\xspace}}

\newcommand{\infix}[3]{\ensuremath{#1(#2,#3)}\xspace}
\newcommand{\suffix}[2]{\infix{#1}{#2}{|#1|+1}}

\newcommand{\wpos}[2]{\ensuremath{#1[#2]}\xspace}

\begin{document}

\maketitle

\begin{abstract}
  We study an operator on classes of languages. For each class~\Cs, it produces a new class~\fodc associated with a variant of two-variable first-order logic equipped with a signature \infsigc built from~\Cs. For $\Cs = \{\emptyset, A^*\}$, we obtain the usual \fodw logic, equipped with linear order.  For $\Cs = \{\emptyset, \{\varepsilon\},A^+, A^*\}$, we get the variant \fodws, which also includes the successor predicate. If \Cs consists of all Boolean combinations of languages $A^*aA^*$, where $a$ is a letter, we get the variant $\fod(<,\mathit{Bet})$, which includes ``between'' relations. We prove a \emph{generic} algebraic characterization of the classes \fodc. It elegantly generalizes those known for all the cases mentioned above. Moreover, it implies that if \Cs has decidable separation (plus some standard properties), then \fodc \mbox{has a decidable membership problem.}

  We actually work with an equivalent definition of \fodc in terms of \emph{unary temporal logic}.  For each class \Cs, we consider a variant \tlc{\Cs} of unary temporal logic whose future/past modalities depend on~\Cs and such that $\tlc{\Cs}= \fodc$. Finally, we also characterize \tlfoc{\Cs} and \tlpoc{\Cs}, the pure-future and pure-past restrictions of \tlc{\Cs}. Like for \tlc\Cs, these characterizations imply that if \Cs is a class with decidable separation, then \tlfoc{\Cs} and \tlpoc{\Cs} have decidable membership.
\end{abstract}

\section{Introduction}
\label{sec:intro}
\textbf{\textsf{Context.}} Regular languages of finite words form a robust class: they admit a wide variety of equivalent definitions, whether by regular expressions, finite automata, finite monoids or monadic second-order logic. It is therefore natural to study the \emph{fragments} of regular languages obtained by restricting the syntax of one of the above-mentioned formalisms. For each particular fragment, we seek to prove that it has a decidable \emph{membership~problem}: given a regular language as input, decide whether it belongs to the fragment.  Intuitively, doing so requires a thorough knowledge of the fragment and the languages it can describe.

This approach was initiated by Schützenberger~\cite{schutzsf} for the class \sfr of star-free languages. These are the languages defined by a \emph{star-free expression}: a regular expression without Kleene star but with complement instead. Equivalently, these are the languages that can be defined in first-order logic with the linear order~\cite{mnpfosf} (\fow) or in linear temporal logic~\cite{kltl}~(\ltl). Schützenberger established an algebraic characterization of \sfr: a regular language is star-free if and only if its \emph{syntactic monoid is aperiodic}. This yields a membership procedure for \sfr because the syntactic monoid can be computed and aperiodicity is a decidable property.

\smallskip
\noindent
\textbf{\textsf{Operators.}} This seminal result prompted researchers to look at other natural classes,~spawning a fruitful line of research (see \emph{e.g.},~\cite{bslt73,simonthm,knast83,pwdelta2,twfo2,pzjacm19}). Although there are numerous classes, they can be grouped into families based on ``variants'' of the same syntax. Let~us use logic to clarify this point. Each logical fragment can use several \emph{signatures} (\emph{i.e.}, sets~of predicates allowed in formulas), each giving rise to a class. For instance, first-order logic is commonly equipped with predicates such as the linear order ``$<$''~\cite{mnpfosf,schutzsf}, the successor~``$+1$''~\cite{Beauquier_1991} or the modular predicates ``$\mathit{MOD}$''~\cite{MIXBARRINGTON1992478}. While it is worth looking at multiple variants of prominent classes, doing so individually for each of them has an obvious disadvantage: the proof has to be systematically modified to accommodate each change. This can be tedious, difficult, and not necessarily enlightening. To overcome this drawback, a natural approach is to capture a whole family of variants with an \emph{operator}. An operator ``\opo'' takes a class~\Cs as input, and outputs a larger one $\opo(\Cs)$. Thus, we can study all classes $\opo(\Cs)$ \emph{simultaneously}: the question becomes: ``what hypotheses about \Cs guarantee the decidability of $\opo(\Cs)$-membership?''. For example, one can generalize the three definitions of star-free languages through~operators:
\begin{enumerate}
  \item The \emph{star-free closure} $\Cs \mapsto \sfp{\Cs}$ has been introduced in~\cite{schutzbd,STRAUBING1979319}. Languages in \sfp{\Cs} are defined by ``extended'' star-free expressions, which can freely use languages from~\Cs.
  \item A construction associating a signature \infsigc to a class \Cs has been given in~\cite{PZ:generic18}. For~each~\mbox{$L \in \Cs$}, the set \infsigc contains a binary predicate $I_L(x,y)$: for a word $w$ and two positions $i,j$ in $w$, $I_L(i,j)$ holds if and only if $i < j$ and the infix of $w$ between $i$ and $j$ belongs to $L$. We get an operator $\Cs \mapsto \foc$ based on first-order logic. It captures \mbox{many choices of signature.}
  \item Similarly, an operator $\Cs \mapsto \ltlc{\Cs}$ that generalizes \ltl has been defined in~\cite{pzsfclosj}.
\end{enumerate}
It is shown in~\cite{PZ:generic18,pzsfclosj} that $\sfp{\Cs} = \foc = \ltlc{\Cs}$ for any class \Cs (with mild hypotheses). Moreover, a \emph{generic} algebraic characterization is proved in~\cite{pzsfclos,pzsfclosj}. Given a regular language~$L$, it relies on a construction that identifies monoids inside its syntactic monoid, called the \emph{\Cs-orbits}: $L \in \sfp{\Cs}$ if and only if its \Cs-orbits are all \emph{aperiodic}. This elegantly generalizes Schützenberger's theorem and gives a \emph{transfer theorem} for membership. Indeed, the \Cs-orbits are connected with a decision problem that strengthens membership: \emph{\Cs-separation}. Given \emph{two} input regular languages $L_1$ and $L_2$, \Cs-separation asks whether there is  $K \in \Cs$ such that \mbox{$L_1 \subseteq K$} and $L_2 \cap K = \emptyset$. The crucial point is that \Cs-orbits are computable if \Cs-separation is decidable. Thus, \sfp{\Cs}-membership is also decidable in this case. Similar results are known for other operators such as \emph{polynomial closure}~\cite{PZ:generic18} or its \emph{unambiguous} restriction~\cite{pzupolicalp,pzupol2}.

\smallskip
\noindent
\textbf{\textsf{Unary temporal logic and two-variable first-order logic.}} The operator we investigate generalizes another important class admitting multiple definitions~\cite{Tesson02diamondsare,small-fragments} (see \cite{fragments-infinite,fragments-traces} for extensions). We are interested in two of them. It consists of languages that can be defined in \emph{two-variable first-order logic} with the linear order (\fodw) or  equivalently in \emph{unary temporal logic} (\tls) with the modalities $\textup{F}$ (sometimes in the future) and $\textup{P}$ (sometimes in the past). Etessami, Vardi and Wilke~\cite{evwutl} have shown that $\fodw=\tls$. Its algebraic characterization  by Thérien and Wilke~\cite{twfo2} is one of the famous results of this type: a regular language belongs to $\fodw=\tls$ if and only if its syntactic monoid belongs to the variety of monoids~\davar.

Both definitions extend to natural operators. First, the generic signatures~\infsigc yield an operator $\Cs \mapsto \fodc$ based on \fod. Second, an operator $\Cs \mapsto \tlc{\Cs}$ has been~defined in~\cite{pzupol2}. It enriches \tls with new modalities $\textup{F}_L$ and $\textup{P}_L$, both depending on the languages~$L\in \Cs$. For example, the formula $\finallyl{\varphi}$ holds at a position $i$ in a word $w$ if there is a position~$j > i$ in $w$ such that $\varphi$ holds at $j$ and the infix between~$i$~and~$j$ belongs to $L$. We know that $\fodc=\tlc{\Cs}$ when \Cs is closed under Boolean operations~\cite{pzupol2}. Here, we work with the $\tlc{\cdot}$ operator, which encompasses all classic classes based on two-variable first-order logic or unary temporal logic. This includes the original variants $\fodw = \tls$ and \mbox{$\fodws = \tlxs$}, both of which were studied by Thérien and~Wilke~\cite{twfo2} (here, ``$+1$'' is the successor predicate and \tlxs is defined by enriching \tls with ``next'' and ``yesterday''  modalities). Another example is the variant \fodwm endowed with \emph{modular predicates}, investigated by Dartois and Paperman~\cite{DartoisP13}. Finally, we capture the variant $\fod(<,Bet) = \textup{BInvTL}$ equipped with ``\emph{between}''  relations, defined and characterized by Krebs, Lodaya, Pandya and Straubing~\cite{betweenlics,betweenconf,between}.

\smallskip
\noindent
\textbf{\textsf{Contributions.}} We prove a \emph{generic} algebraic characterization of the classes $\fodc=\tlc{\Cs}$. We reuse the \Cs-orbits introduced for star-free closure: for any class \Cs (having mild properties) we show that a regular language belongs to \tlc{\Cs} if and only if all \Cs-orbits of its syntactic monoid belong to \davar. In particular, this yields a \emph{transfer theorem} for membership: if \Cs has decidable \emph{separation}, then $\fodc=\tlc{\Cs}$ has decidable membership. Moreover, this characterization generalizes the characterizations known for all the above instances.

A key feature of our proof is that we use a third auxiliary operator. It combines two other operators: Boolean polynomial closure (\bpolo) and \mbox{unambiguous polynomial closure~(\upolo)}. We have $\upol{\bpol{\Cs}} \subseteq \tlc{\Cs}$ if \Cs has mild properties~\cite{pzupol2}. In fact, for many natural classes, \mbox{$\upol{\bpol{\Cs}}\,{=}\, \tlc{\Cs}$}. For example, $\upol{\bpol{\{\emptyset,A^*\}}}$ is the class \ul of~\emph{unambiguous languages} defined by Schützenberger~\cite{schul}. It is known~\cite{twfo2} that $\ul =\tls=\fodw$. More generally, $\upol{\bpol{\Cs}} = \tlc{\Cs}$ for every class \Cs consisting of \emph{group languages}~\cite{pzupol2}. Yet, this is a strong hypothesis and the inclusion $\upol{\bpol{\Cs}} \subseteq \tlc{\Cs}$ is \emph{strict} in general. For example, the results of~\cite{between} imply that $\upol{\bpol{\at}} \neq \tlc{\at}$, where \at consists of all Boolean combinations of languages $A^*aA^*$ (with $a \in A$). Nevertheless, the classes \upol{\bpol{\Cs}} serve as a central ingredient in the most difficult direction of our proof: ``\emph{If a language satisfies our characterization on \Cs-orbits, prove that it belongs to \tlc{\Cs}}''. More precisely, we exploit the known characterization of \upol{\bpol{\Cs}} to prove that auxiliary languages belong to this class, and we then conclude using the \mbox{inclusion $\upol{\bpol{\Cs}} \subseteq \tlc{\Cs}$}.

Finally, we look at two additional operators: $\Cs \mapsto \tlfoc{\Cs}$ and $\Cs \mapsto \tlpoc{\Cs}$. They are also defined in terms of unary temporal logic, as the pure-future and the pure-past restrictions of $\Cs \mapsto \tlc{\Cs}$. We present generic algebraic characterizations for these two operators as well. Again, they are based on \Cs-orbits. For every class \Cs (with mild hypotheses), we show that a regular language belongs to \tlfoc{\Cs} (resp.~\tlpoc{\Cs}) if and only if all the \Cs-orbits inside its syntactic monoid are \Lrel-trivial (resp. \Rrel-trivial) monoids. As before,  these results yield \emph{transfer theorems}: if \Cs has decidable \emph{separation}, then \tlfoc{\Cs} and \tlpoc{\Cs} have \mbox{decidable membership.}

\smallskip\noindent\textbf{\textsf{Organization of the paper.}} We recall the notation and background in Section~\ref{sec:prelims}. In Section~\ref{sec:orbits}, we present the \Cs-orbits and their properties. In Section~\ref{sec:utl}, we define the operator~$\Cs\mapsto\tlc\Cs$. Section~\ref{sec:utlcar} is devoted to the generic characterization of \tlc\Cs and to its proof. In Section~\ref{sec:fpcar}, finally, we state the characterizations of the pure-future and pure-past restrictions of~$\tlc\Cs$.

\section{Preliminaries}
\label{sec:prelims}
We fix a finite \emph{alphabet} $A$ for the paper. As usual, $A^*$ denotes the set of all finite words over~$A$, including the empty word \veps. A \emph{language} is a subset of $A^*$. We let $A^+ = A^* \setminus \{\veps\}$. For $u,v \in A^*$, we write~$uv$ for the word obtained by concatenating $u$ and $v$. We lift the concatenation to languages as follows: if $K,L\subseteq A^*$, we let $KL=\{uv\mid u\in K, v\in L\}$. If $w \in A^*$, we write $|w| \in \nat$ for its length. A word $w =a_1 \cdots a_{|w|} \in A^*$ is viewed as an \emph{ordered set $\pos{w} = \{0,1,\dots,|w|,|w|+1\}$ of $|w|+2$ positions}. In addition, we let $\posc{w} = \{1,\dots,|w|\} \subsetneq \pos{w}$. Position $i \in \posc{w}$ carries label $a_i \in A$, which we write $\wpos{w}{i} = a_i$. On the other hand, positions $0$ and $|w|+1$ carry \emph{no label}. We write $\wpos{w}{0} = min$ and $\wpos{w}{|w|+1} = max$. For $v,w \in A^*$, we say that $v$ is an \emph{infix} (resp.~\emph{prefix}, \emph{suffix}) of~$w$ when there exist $x,y \in A^*$ such that $w = xvy$ (resp.~$w=vy$, $w=xv$). Given a word $w = a_1\cdots a_{|w|} \in A^*$ and $i,j \in \pos{w}$ such that $i < j$, we write $\infix{w}{i}{j} = a_{i+1} \cdots a_{j-1} \in A^*$ (\emph{i.e.}, the infix obtained by keeping the letters carried by positions \emph{strictly} between $i$ and $j$).

\smallskip
\noindent
\textbf{\textsf{Classes}.}
 A \emph{class} of languages \Cs is simply a set of languages. Such a class \Cs is a \emph{lattice} when $\emptyset\in\Cs$, $A^* \in \Cs$ and \Cs is closed under both union and intersection: for all $K,L \in \Cs$, we have $K \cup L \in \Cs$ and $K \cap L \in \Cs$. Moreover, a class of languages \Cs is a \emph{Boolean algebra} if it is a lattice closed under complement: for all $L \in \Cs$, we have $A^* \setminus L \in \Cs$. Finally, the class \Cs is \emph{closed under quotients} if for all $L \in \Cs$ and $u \in A^*$, we have $u^{-1}L \stackrel{\smash{\text{\tiny def}}}= \{w\in A^*\mid uw\in L\} \in \Cs$ and $Lu^{-1}  \stackrel{\smash{\text{\tiny def}}}=\{w\in A^*\mid wu\in L\}\in \Cs$. A \emph{\vari} is a Boolean algebra closed under quotients and containing only \emph{regular languages}. Regular languages are those which can be equivalently defined by finite automata, finite monoids or monadic second-order logic. We work with the definition by monoids, which we now recall.

\smallskip
\noindent
\textbf{\textsf{Monoids}.} A \emph{monoid} is a set $M$  endowed with an associative multiplication $(s,t)\mapsto st$ having an identity element $1_M$ (\emph{i.e.}, such that ${1_M}s=s{1_M}=s$ for every~$s \in M$). An \emph{idempotent} of a monoid $M$ is an element $e \in M$ such that $ee = e$. We write $E(M) \subseteq M$ for the set of all idempotents in $M$. It is folklore that for every \emph{finite} monoid~$M$, there exists a natural number $\omega(M)$ (denoted by $\omega$ when $M$ is understood) such that for every $s \in M$, the element $s^\omega$ is an idempotent.
Finally, we shall use the following Green relations~\cite{green} defined on monoids. Given a monoid $M$ and $s,t \in M$, we write:
\[
  \begin{array}{ll@{\ }l}
	s \Jord t & \text{when} & \text{there exist $x,y\in M$ such that $s=xty$}, \\
	s \Lord t & \text{when} & \text{there exists $x \in M$ such that $s = xt$}, \\
	s \Rord t & \text{when} & \text{there exists $y \in M$ such that $s = ty$}. \\
  \end{array}
\]
Clearly, \Jord, \Lord and \Rord are preorders (\emph{i.e.}, they are reflexive and transitive). We write \Jords, \Lords and \Rords for their strict variants (for example, $s \Jords t$ when $s \Jord t$ but $t \not\Jord s$). Finally, we write \Jrel, \Lrel and \Rrel for the corresponding equivalence relations (for example, $s \Jrel t$ when $s \Jord t$ and $t \Jord s$). There are many technical results about Green relations. We will just need the following easy and standard lemma, which applies to \emph{finite} monoids (see \emph{e.g.},~\cite{pinvars,pingoodref}).

\begin{restatable}{lemma}{jlr} \label{lem:jlr}
  Let $M$ be  a finite monoid and let $s,t \in M$. If $s \Jrel t$ and $s \Rord t$, then $s \Rrel t$. \end{restatable}

\noindent
\textbf{\textsf{Regular languages and syntactic morphisms}.} Since $A^*$ is a monoid whose multiplication is concatenation (the identity element is \veps), we may consider monoid morphisms $\alpha: A^* \to M$ where $M$ is an arbitrary monoid. That is, $\alpha:A^*\to M$ is a map satisfying $\alpha(\veps)=1_M$ and $\alpha(uv)=\alpha(u)\alpha(v)$ for all $u,v\in A^*$. We say that a language $L \subseteq A^*$ is \emph{recognized} by $\alpha$ when there exists a set $F \subseteq M$ such that $L= \alpha\inv(F)$.

It is well known that regular languages are exactly those recognized by a morphism into a \emph{finite} monoid. Moreover, every language $L$ is recognized by a canonical morphism, which we briefly recall. One can associate to $L$ an equivalence $\equiv_L$ over $A^*$: the \emph{syntactic congruence of~$L$}. Given $u,v \in A^*$, we let $u \equiv_L v$ if and only if $xuy \in L \Leftrightarrow xvy \in L$ for every $x,y \in A^*$.
One can check that ``$\equiv_L$'' is indeed a congruence on $A^*$: it is an equivalence compatible with word concatenation. Thus, the set of equivalence classes $M_L = {A^*}/{\equiv_L}$ is a monoid. It is called the \emph{syntactic monoid of $L$}. Finally, the map $\alpha_L: A^* \to M_L$ sending every word to its equivalence class is a morphism recognizing~$L$, called the \emph{syntactic morphism of~$L$}. It is known that a language~$L$ is regular if and only if $M_L$ is finite (\emph{i.e.}, $\equiv_L$ has finite index): this is the Myhill-Nerode theorem. In this case, one can compute the syntactic morphism $\alpha_L: A^* \to M_L$ from any representation of $L$ (such as an automaton or a monoid morphism).

\medskip
\noindent
\textbf{\textsf{Decision problems}.} We consider two decision problems, both depending on an arbitrary class \Cs. They serve as mathematical tools for analyzing it, as obtaining an algorithm for one of these problems requires a solid understanding of that class~\Cs.
The \emph{\Cs-membership problem} is the simplest: it takes as input a single regular language~$L$ and simply asks whether $L\in \Cs$. The second problem, \emph{\Cs-separation}, is more general. Given three languages $K,L_1,L_2$, we say that $K$ \emph{separates} $L_1$ from $L_2$ if $L_1 \subseteq K$ and $L_2 \cap K = \emptyset$. Given a class~\Cs, we say that $L_1$ is \emph{\Cs-separable} from $L_2$ if some language of \Cs separates $L_1$ from $L_2$. The \emph{\Cs-separation problem} takes as input two regular languages $L_1,L_2$ and asks whether $L_1$ is \Cs-separable from $L_2$.

\begin{remark} \label{rem:sepgenmem}
  The \Cs-separation problem generalizes \Cs-membership. Indeed, a regular language belongs to $\Cs$ if and only if it is \Cs-separable from its complement, which is regular.
\end{remark}

\section{Orbits}
\label{sec:orbits}
Instead of looking at single classes, we consider \emph{operators}. These are correspondences $\Cs\mapsto \opo(\Cs)$ that take as input a class \Cs to build a new one $\opo(\Cs)$. We investigate three operators in Sections~\ref{sec:utl} to \ref{sec:fpcar}. For now, we present general tools for handling such operators. Given a class \Cs and a morphism $\alpha: A^* \to M$, we define special subsets of $M$: the \Cs-orbits for~$\alpha$. This notion was introduced in~\cite{pzsfclosj}. We shall use it to formulate \emph{generic characterizations} of the operators $\Cs\mapsto \opo(\Cs)$ that we consider: for each input \vari \Cs, the languages in $\opo(\Cs)$ are characterized by a property of the \Cs-orbits for their syntactic~morphisms.

\smallskip
\noindent
\textbf{\textsf{\boldmath{\Cs}-pairs}.} Consider a class \Cs and a morphism $\alpha: A^* \to M$. We say that a pair $(s,t) \in M^2$  is a \emph{\Cs-pair for $\alpha$} if and only if $\alpha\inv(s)$ is \emph{not} \Cs-separable from $\alpha\inv(t)$. Note that if \Cs-separation is decidable, then one can compute all \Cs-pairs for an input morphism.

We turn to a useful technical result, which characterizes the \Cs-pairs using morphisms. Consider \emph{two} morphisms $\alpha: A^*\to M$ and $\eta: A^* \to N$. For every pair $(s,t) \in M^2$, we say that $(s,t)$ is an \emph{$\eta$-pair for $\alpha$} when there exist $u,v \in A^*$ such that $\eta(u) = \eta(v)$, $\alpha(u) = s$ and $\alpha(v) = t$. In addition, for each class \Cs, we define the \emph{\Cs-morphisms} as the \emph{surjective} morphisms $\eta: A^*\to N$ into a finite monoid $N$ such that all languages recognized by $\eta$ belong to \Cs. We have the following elementary lemma, proved in~\cite[Lemma 5.11]{pzupol2}.

\begin{restatable}{lemma}{pmor}\label{lem:pairmor}
  Let \Cs be a \vari and $\alpha: A^* \to M$ be a morphism. Then,
  \begin{enumerate}
    \item For every \Cs-morphism $\eta: A^* \to N$, all \Cs-pairs for $\alpha$ are also $\eta$-pairs for $\alpha$.
    \item There exists a \Cs-morphism $\eta: A^* \to N$ such that all $\eta$-pairs for $\alpha$ are also $\Cs$-pairs for $\alpha$.
  \end{enumerate}

\end{restatable}

\smallskip\noindent
\textbf{\textsf{\boldmath{\Cs}-orbits}.} Consider a class \Cs and a morphism $\alpha: A^* \to M$. For every \emph{idempotent} $e \in E(M)$, the \emph{\Cs-orbit of $e$ for $\alpha$} is the set $M_e \subseteq M$ consisting of all elements $ete \in M$ such that $(e,t) \in M^2$ is a \Cs-pair. If \Cs is a \vari and $\alpha$ is surjective, it is proved in~\cite[Lemma~5.5]{pzsfclosj} that $M_e$ is a monoid in $M$: it is closed under multiplication and $e \in M_e$ is its identity. On the other hand, $M_e$ is not a ``\emph{sub}monoid'' of $M$ (this is because $1_M$ needs not belong to $M_e$).

\begin{restatable}{lemma}{orbitmono} \label{lem:orbitmono}
  Let \Cs be a \vari and $\alpha: A^* \to M$ be a surjective morphism into a finite monoid. For all $e \in E(M)$, the \Cs-orbit of $e$ for $\alpha$ is a monoid in $M$ whose identity is $e$.
\end{restatable}

As seen above, when \Cs has decidable separation, one can compute the \Cs-pairs associated with an input morphism. Hence, one can also compute the \Cs-orbits in this case.

\begin{restatable}{lemma}{orbitcomp} \label{lem:orbitcomp}
  Let \Cs be a class with decidable separation. Given as input a morphism $\alpha: A^* \to M$ into a finite monoid and $e \in E(M)$, one can compute the \Cs-orbit of $e$ for $\alpha$.
\end{restatable}

Finally, the following  lemma connects \Cs-orbits with \Cs-morphisms.

\begin{restatable}{lemma}{orbitnec} \label{lem:orbitnec}
  Let \Cs be a \vari and $\alpha: A^* \to M$ be a morphism. Moreover, let $\eta: A^* \to N$ be a \Cs-morphism. For every $e \in E(M)$, there exists $f \in E(N)$ such that the \Cs-orbit of $e$ for~$\alpha$ is contained in the set $\alpha(\eta\inv(f))$.
\end{restatable}

\begin{proof}
  Let $t_1,\dots,t_n \in M$ be all elements of the set $\{t\in M\mid (e,t)\text{ is a \Cs-pair}\}$. By definition, the \Cs-orbit of $e$ for $\alpha$ is $M_e = \{et_1e,\dots,et_ne\}$. Since $\eta$ is a \Cs-morphism, Lemma~\ref{lem:pairmor} implies that $(e,t_i)$ is an $\eta$-pair for all $i \leq n$. This yields $x_i,y_i \in A^*$ such that $\eta(x_i) = \eta(y_i)$, $\alpha(x_i) = e$ and $\alpha(y_i) = t_i$. Let $p = \omega(N)$, $w = (x_1 \cdots x_n)^p$ and $f = \eta(w)$. Note that~$f$ is idempotent by choice of $p$. We show that $et_ie \in \alpha(\eta\inv(f))$ for $i \leq n$. We define $w_i = (x_1 \cdots x_n)^p x_1 \cdots x_{i-1} y_i x_{i+1} \cdots x_n (x_1 \cdots x_n)^{2p-1}$. By definition, we have $\alpha(w_i) = et_ie$. Now, since $\eta(x_i) = \eta(y_i)$, we get $\eta(w_i) = \eta(w) = f$. Hence, $et_ie \in \alpha(\eta\inv(f))$, as desired.
\end{proof}

\section{Generalized unary temporal logic}
\label{sec:utl}
In this section, we define generalized unary temporal logic. We introduce an \emph{operator} $\Cs\mapsto\tlc{\Cs}$ that associates a new class of languages \tlc{\Cs} with every input class~\Cs. We first recall its definition (taken from~\cite{pzupol2}), and we then complete it with useful properties.

\subsection{Definition}

\noindent
\textbf{\textsf{Syntax}.} We associate with any class \Cs a set of temporal formulas denoted by \tla{\Cs} as follows. A \tla{\Cs} formula is built from atomic formulas using Boolean connectives and temporal operators. The atomic formulas are $\top$, $\bot$, $min$, $max$ and ``$a$'' for every letter $a \in A$. All Boolean connectives are allowed: if $\psi_1$ and $\psi_2$ are \tla{\Cs} formulas, then so are $(\psi_{1} \vee \psi_{2})$, $(\psi_{1} \wedge \psi_{2})$ and $(\neg \psi_1)$. We associate \emph{two temporal modalities} with every language $L \in \Cs$, which we denote by $\textup{F}_L$ and $\textup{P}_L$: if $\psi$ is a \tla{\Cs} formula, then so are $(\finallyl{\psi})$ and $(\finallyml{\psi})$. For the sake of improved readability, we omit parentheses when there is no ambiguity.

\smallskip
\noindent
\textbf{\textsf{Semantics}.} Evaluating a \tla{\Cs} formula $\varphi$ requires a word $w \in A^*$ and a position $i \in \pos{w}$. We define by induction what it means for \emph{$(w,i)$ to satisfy $\varphi$}, which one denotes by $w,i \models \varphi$.
\begin{itemize}
  \item \textbf{Atomic formulas:} $w,i \models \top$ always holds, $w,i \models \bot$ never holds and for every symbol $\ell \in A \cup \{min,max\}$, $w,i \models \ell$ holds when $\ell = \wpos{w}{i}$.
  \item \textbf{Disjunction:} $w,i \models \psi_1 \vee \psi_2$ when $w,i \models \psi_1$ or $w,i \models \psi_2$.
  \item \textbf{Conjunction:} $w,i \models \psi_1 \wedge \psi_2$ when $w,i \models \psi_1$ and $w,i \models \psi_2$.
  \item \textbf{Negation:} $w,i \models \neg \psi$ when $w,i \models \psi$ does not hold.
  \item \textbf{Finally:}  for $L \in \Cs$, we let $w,i \models \finallyl{\psi}$ when there exists $j \in \pos{w}$ such that $i < j$,   $\infix{w}{i}{j} \in L$ and $w,j \models \psi$.
  \item \textbf{Previously:} for $L \in \Cs$, we let $w,i \models \finallyml{\psi}$ when there exists $j \in \pos{w}$ such that $j < i$,   $\infix{w}{j}{i} \in L$ and $w,j \models \psi$.
\end{itemize}
When no distinguished position is specified, it is customary to evaluate formulas at the~\emph{leftmost} unlabeled position. One could also consider the symmetrical convention of evaluating formulas at the \emph{rightmost} unlabeled position. The convention chosen does not matter: we end-up with the same class of languages. However, we shall consider restrictions of \tla{\Cs} for which this choice \emph{does} matter. This is why we introduce notations for both conventions. Given a formula $\varphi\in\tla{\Cs}$ we let $L_{min}(\varphi) = \{w \in A^* \mid w,0 \models \varphi\}$ and $L_{max}(\varphi) = \{w \in A^* \mid w,|w|+1 \models \varphi\}$.

We are now ready to define the operator $\Cs\mapsto\tlc\Cs$. Consider an arbitrary class \Cs. We write \tlc{\Cs} for the class consisting of all languages $L_{min}(\varphi)$ where $\varphi \in \tla{\Cs}$. Observe that by definition, \tlc{\Cs} is a Boolean algebra. Actually, the results of~\cite{pzupol2} imply that when \Cs is a \vari, then so is \tlc{\Cs} (we do not need this fact in the present paper).

\medskip
\noindent
\textbf{\textsf{Classic unary temporal logic}.} Let $\stzer =\{\emptyset,A^*\}$ and $\dotzer = \{\emptyset,\{\veps\}, A^+,A^*\}$. The modalities $\textup{F}_{A^*}$ and $\textup{P}_{A^*}$ have the same semantics as the modalities $\textup{F}$ and $\textup{P}$ of standard unary temporal logic---\emph{e.g.}, $w,i \models \finally{\varphi}$ when there exists $j \in \pos{w}$ such that $i < j$ and $w,j \models \varphi$. Similarly, the modalities $\textup{F}_{\{\veps\}}$ and $\textup{P}_{\{\veps\}}$ have the same semantics as the modalities $\textup{X}$ (next) and $\textup{Y}$ (yesterday)---\emph{e.g.}, $w,i \models \nex{\varphi}$ when $i+1 \in \pos{w}$ and $w,i+1 \models \varphi$. Using these facts, one can check that the classes $\tlc{\stzer}$ and $\tlc{\dotzer}$ correspond exactly to the two original standard variants of unary temporal logic (see \emph{e.g.},~\cite{evwutl}): we have $\tls = \tlc{\stzer}$ and $\tlxs = \tlc{\dotzer}$.

\begin{remark}[Robustness of classes to which \tls is applied] \label{rem:robustness}
  Note that including $\emptyset$ in an input class does not bring any new modality in unary temporal logic. Similarly, the classes \tlc\dotzer and $\tlc{\dotzer\setminus \{A^+\}}$ are identical. However, in order to use generic results such as those from Section~\ref{sec:orbits}, we require the classes to which the operator $\Cs\mapsto\tlc\Cs$ is applied to have robust properties: they should be \varis (hence, they should be closed under~complement).
\end{remark}

\begin{remark}[Connection with \fod]\label{rem:fo2}
  Etessami, Vardi and Wilke~\cite{evwutl} have shown that the variant~\tls corresponds to the class \fodw (two-variable first-order logic equipped with the linear order), and that \tlxs corresponds to \fodws (which also allows the successor). In~\cite{pzupol2}, these results are generalized to all classes \tlc{\Cs} where \Cs is~a \emph{Boolean algebra}. In this case, we can construct from \Cs a set of predicates \infsigc such that $\tlc{\Cs} = \fod(\infsigc)$.
\end{remark}

\begin{remark} \label{rem:at}
  Another important input is the class \at of alphabet testable languages. It consists of all Boolean combinations of languages $A^*aA^*$, where $a \in A$ is a letter. The class \tlc{\at} has been studied by Krebs, Lodaya, Pandya and Straubing~\cite{betweenlics,betweenconf,between}, who worked with the definition based on two-variable first-order logic (\emph{i.e.}, with the class $\fod(\infsig{\at})$, see Remark~\ref{rem:fo2}). In particular, they proved that \tlc{\at} has decidable membership. We shall obtain this result as a corollary of our generic characterization of the classes \tlc{\Cs}.
\end{remark}

\subsection{Connection with unambiguous polynomial closure}

It is shown in~\cite{pzupol2} that $\Cs \mapsto \tlc{\Cs}$ can be expressed by other operators for very specific inputs: \varis of \emph{group languages}. If \Gs is such a class, then $\tlc{\Gs}$ coincides with $\upol{\bpol{\Gs}}$, a class built on top of \Gs with the two standard operators \upolo and \bpoln. We do not use this result here, since we are tackling \emph{arbitrary} input \varis, and in general, \upol{\bpol{\Cs}} is \emph{strictly} included in \tlc{\Cs} (it follows from~\cite{between} that the inclusion is strict for the class $\at$ of Remark~\ref{rem:at}). However, the operators \upolo and \bpolo remain key tools in the paper: we use two results of~\cite{pzupol2} about them. Let us first briefly recall their~definitions.

Given finitely many languages $L_0,\dots,L_n \subseteq A^*$, a \emph{marked product} of $L_0,\dots,L_n$ is a~product of the form $L_0a_1L_1 \cdots a_n L_n$ where $a_1,\dots,a_n \in A$. A single language $L_0$ is a marked product (this is the case $n = 0$). The \emph{polynomial closure} of a class \Cs, denoted by \pol{\Cs}, consists of all \mbox{\emph{finite}} \emph{unions} of marked products $L_0a_1L_1\cdots a_nL_n$ such that $L_0,\dots, L_n \in \Cs$. If \Cs is a \vari, then \pol{\Cs} is a lattice (this is due to Arfi~\cite{arfi87}, see also~\cite{jep-intersectPOL,PZ:generic18} for recent proofs). However, \pol{\Cs} need not be closed under complement. This is why it is often combined with another operator: the Boolean closure of a class \Ds, denoted by \bool{\Ds}, is the least Boolean algebra containing \Ds. We write \bpol{\Cs} for \bool{\pol{\Cs}}. It is standard that if \Cs is a \vari,  then so is \bpol{\Cs} (see~\cite{PZ:generic18} for example). Finally, \upolo is the \emph{unambiguous} restriction of \poln. A marked product $L_0a_1L_1\cdots a_nL_n$ is \emph{unambiguous} when every word $w \in L_0a_1L_1 \cdots a_n L_n$ has a \emph{unique} decomposition $w = w_0a_1w_1 \cdots a_nw_n$ where $w_i \in L_i$ for $0 \leq i \leq n$. The \emph{unambiguous polynomial closure} of a class \Cs, written \upol{\Cs}, consists of all \emph{finite disjoint unions} of \emph{unambiguous marked products} $L_0a_1L_1 \cdots a_nL_n$ such that $L_0, \dots,L_n \in \Cs$ (by ``disjoint'' we mean that the languages in the union must be pairwise disjoint). While this is not apparent on the definition, it is known~\cite{pzupol2} that if the input class~\Cs is a \vari, then so is~\upol{\Cs}. Thus, \upolo preserves closure under complement.

\smallskip
In the paper, we are interested in the ``combined'' operator $\Cs \mapsto \upol{\bpol{\Cs}}$. Indeed, it is connected to the classes \tlc{\Cs} by the following proposition proved in~\cite[Proposition~9.12]{pzupol2}.

\begin{restatable}{proposition}{ubptl} \label{prop:ubptl}
  For every \vari \Cs, we have $\upol{\bpol{\Cs}} \subseteq \tlc{\Cs}$.
\end{restatable}

Although the inclusion of Proposition~\ref{prop:ubptl} is \emph{strict} in general, it is essential for proving that \emph{particular} languages belong to \tlc{\Cs}. Indeed, we will combine it with the next result~\cite[Theorem~6.7]{pzupol2} to prove that languages belong to \upol{\bpol{\Cs}}---and therefore to~\tlc\Cs.

\begin{restatable}{theorem}{ubpcar}\label{thm:ubpcar}
  Let \Cs be a \vari, $L \subseteq A^*$ be a regular language and $\alpha: A^* \to M$ be its syntactic morphism. Then, $L \in \upol{\bpol{\Cs}}$ if and only if $\alpha$ satisfies the following~property:
  \begin{equation} \label{eq:ubp}
	(esete)^{\omega+1} = (esete)^{\omega} ete(esete)^{\omega} \quad \text{for every \Cs-pair $(e,s) \in M^2$ and every $t \in M$}.
  \end{equation}
\end{restatable}

\section{Algebraic characterization of \tlc{\Cs}}
\label{sec:utlcar}
We present a generic characterization of \tlc{\Cs} when \Cs is a \vari. It elegantly generalizes the characterizations of $\tls = \fodw$ and $\tlxs = \fodws$ by Thérien and Wilke~\cite{twfo2} and that of $\tlc{\at} = \fod(\infsig{\at})$ by Krebs, Lodaya, Pandya and Straubing~\cite{betweenlics,betweenconf,between}.

\subsection{Statement}

The characterization is based on the well-known variety of finite monoids \davar (see~\cite{Tesson02diamondsare} for a survey on this class). A finite monoid $M$ belongs to \davar if it satisfies the following~equation:
\begin{equation} \label{eq:da}
  (st)^\omega = (st)^\omega t (st)^\omega \quad \text{for every $s,t \in M$}.
\end{equation}
Thérien and Wilke~\cite{twfo2} showed that a regular language belongs to \tls if and only if its syntactic monoid is in~\davar (strictly speaking, they considered two-variable first-order logic, the equality $\fodw = \tls$ is due to Etessami, Vardi and Wilke~\cite{evwutl}). We extend this result in the following generic characterization of \tlc{\Cs}, based on \Cs-orbits introduced in Section~\ref{sec:orbits}.

\begin{restatable}{theorem}{utlmain} \label{thm:utlmain}
  Let \Cs be a \vari, $L \subseteq A^*$ be a regular language and $\alpha: A^* \to M$ be its syntactic morphism. The two following properties are equivalent:
  \begin{enumerate}
    \item $L \in \tlc{\Cs}$.
    \item For every idempotent $e \in E(M)$, the \Cs-orbit of $e$ for $\alpha$ belongs to \davar.
  \end{enumerate}
\end{restatable}

Given as input a regular language $L \subseteq A^*$, one can compute its syntactic morphism $\alpha: A^* \to M$. In view of Theorem~\ref{thm:utlmain}, $L \in \tlc{\Cs}$ if and only if for every $e \in E(M)$, the \Cs-orbit of $e$ for $\alpha$ belongs to \davar. The latter condition can be decided by checking all \Cs-orbits, provided that we are able to compute them. By Lemma~\ref{lem:orbitcomp}, this is possible when \Cs-separation is decidable. Altogether, we obtain the following corollary of Theorem~\ref{thm:utlmain}.

\begin{restatable}{corollary}{cutlmain} \label{cor:utlmain}
  If a \vari \Cs has decidable separation, \tlc{\Cs} has decidable~membership.
\end{restatable}

\begin{remark}\label{rem:upolbpol-vs-tlc}
  Let $L\subseteq A^*$  be a regular language and $\alpha: A^* \to M$ be its syntactic morphism. The fact that the \Cs-orbit of $e\in E(M)$ for $\alpha$ belongs to \davar means that we have,
  \begin{equation}\label{eq:orbit-in-davar}
    (esete)^\omega = (esete)^\omega ete (esete)^\omega \quad \text{for all $s,t \in M$ such that $(e,s)$ and $(e,t)$ are \Cs-pairs.}
  \end{equation}
  One can check that~\eqref{eq:orbit-in-davar} follows from~\eqref{eq:ubp}, which characterizes~$\upol{\bpol{\Cs}}$ (this is consistent with Proposition~\ref{prop:ubptl}  asserting that $\upol{\bpol{\Cs}}\subseteq\tlc\Cs$). Indeed, choosing $t=s$ in~\eqref{eq:ubp} shows that the \Cs-orbit of $e$ is aperiodic, \emph{i.e.}, $(ese)^{{\omega+1}}=(ese)^\omega$ if $(e,s)$ is a \Cs-pair. However, note that the element $t$ is ``free'' in~\eqref{eq:ubp}, whereas it must be part of a \Cs-pair $(e,t)$ in~\eqref{eq:orbit-in-davar}.
\end{remark}

Before proving Theorem~\ref{thm:utlmain}, we first explain why it generalizes the original characterizations of the classes \tls, \tlxs and \tlc{\at}, as mentioned at the beginning of the section.

\subsection{Application to historical classes}

We first deduce the original characterizations of the classes $\tls = \tlc{\stzer}$ and $\tlxs = \tlc{\dotzer}$ by Thérien and Wilke~\cite{twfo2} as simple corollaries of Theorem~\ref{thm:utlmain}. We start with the former.

\begin{restatable}[Thérien and Wilke~\cite{twfo2}]{theorem}{thmtls} \label{thm:tls}
  Let $L \subseteq A^*$ be a regular language and let $M$ be its syntactic monoid. The two following properties are equivalent:
  \begin{enumerate}
    \item $L$ belongs to $\tls$.
    \item $M$ belongs to \davar.
  \end{enumerate}
\end{restatable}

\begin{proof}
  Let $\alpha: A^* \to M$ be the syntactic morphism of $L$. Since $\tls = \tlc{\stzer}$, Theorem~\ref{thm:utlmain} implies that $L \in \tls$ if and only if every \stzer-orbit for $\alpha$ belongs to \davar. Since $\stzer=\{\emptyset,A^*\}$, every pair $(e,s)\in E(M)\times M$ is a \Cs-pair, so that the \stzer-orbit of $e\in E(M)$ for $\alpha$ is $eMe$. In particular the \stzer-orbit of $1_M$ is the whole monoid $M$. Hence, every \stzer-orbit for $\alpha$ belongs to \davar if and only if $M$ belongs to \davar, which completes the proof.
\end{proof}

We turn to the characterization of $\tlxs = \tlc{\dotzer}$, also due to Thérien and Wilke~\cite{twfo2}. In order to state it, we need an additional definition. Consider a regular language~$L$ and let $\alpha: A^* \to M$ be its syntactic morphism. The \emph{syntactic semigroup of $L$} is the set $S = \alpha(A^+)$. Note that for every idempotent $e \in E(S)$, the set $eSe$ is a monoid whose neutral element is $e$.

\begin{restatable}[Thérien and Wilke~\cite{twfo2}]{theorem}{thmtlxs} \label{thm:tlxs}
  Let $L \subseteq A^*$ be a regular language and $S$ be its syntactic semigroup. The two following properties are equivalent:
  \begin{enumerate}
    \item $L$ belongs to $\tlxs$.
    \item For every $e \in E(S)$, the monoid $eSe$ belongs to \davar.
  \end{enumerate}

\end{restatable}

\begin{proof}
  Let $\alpha: A^* \to M$ be the syntactic morphism of $L$. For $e \in E(M)$, let $M_e \subseteq M$ be the \dotzer-orbit of $e$ for $\alpha$. Since $\dotzer = \{\emptyset,\{\veps\},A^+,A^*\}$, for $(e,s)\in E(S)\times S$, the language $\alpha\inv(e)$ is not \dotzer-separable from $\alpha\inv(s)$. Hence, $(e,s)$ is a \Cs-pair, so that $M_e = eSe$ for all $e \in E(S)$. Moreover, if $1_M\not\in E(S)$ (which means that $\alpha\inv(1_M)=\{\veps\}$), then we have $M_{1_M} = \{1_M\}$ (which clearly belongs to \davar). Hence, every \dotzer-orbit for $\alpha$ belongs to \davar if and only if $eSe \in \davar$ for every $e \in E(S)$. In view of Theorem~\ref{thm:utlmain}, this implies Theorem~\ref{thm:tlxs}.
\end{proof}

Finally, we consider the class \tlc{\at}, defined and characterized by Krebs, Lodaya,~Pandya and Straubing~\cite{betweenlics,betweenconf,between}. Let us first present their characterization. It is based on a variety of finite monoids called \meda. Let $M$ be a finite monoid. For each $e \in E(M)$, let $N_e \subseteq M$ be the submonoid of $M$ generated by the set $\{s\in M\mid e\Jord s\}$. We say that $M$ belongs to \meda if and only if for every idempotent $e\in E(M)$, the monoid of $eN_ee$ belongs to \davar.

\begin{restatable}[Krebs, Lodaya, Pandya and Straubing~\cite{between}]{theorem}{thmtlat} \label{thm:tlat}
  Let $L \subseteq A^*$ be a regular language and $M$ be its syntactic monoid. The two following properties are equivalent:
  \begin{enumerate}
    \item $L \in \tlc{\at}$.
    \item $M$ belongs to \meda.
  \end{enumerate}

\end{restatable}

\begin{proof}
  For $w \in A^*$, let $\cont{w} \subseteq A$ be the set of letters occurring in $w$ (\emph{i.e.}, the least set $B \subseteq A$ such that $w \in B^*$). For $e \in E(M)$, let $M_e$ be the \at-orbit of $e$ for $\alpha$. We prove that $M_e = e N_e e$ for every $e \in E(M)$. It will follows that $M$ belongs to \meda if and only if every \at-orbit for $\alpha$ belongs to~\davar. In view of Theorem~\ref{thm:utlmain} this implies Theorem~\ref{thm:tlat}.

  We first consider $s' \in eN_ee$ and prove that $s' \in M_e$. We have $s \in N_e$ such that $s' = ese$. By definition, $s = s_1 \cdots s_n$ where $e \Jord s_i$ for every $i \leq n$. If $n = 0$, then $s = 1_M$ and $ese = e \in M_e$. Assume now that we have \mbox{$n \geq 1$}. Since $e \Jord s_i$, we have $q_i,r_i \in M$ such that $e = q_is_ir_i$ for every $i \leq n$. Hence, since $e \in E(M)$, we have $e = q_1s_1r_1 \cdots q_n s_nr_n$. For every $i \leq n$, let $x_i \in \alpha\inv(q_i)$, $y_i \in \alpha\inv(r_i)$ and $u_i \in \alpha\inv(s_i)$. Finally, let $w = x_1u_1y_1 \cdots x_nu_ny_n$ and $w' = w u_1 \cdots u_n w$. By definition, we have $e = \alpha(w)$ and $ese = \alpha(w')$. Moreover, it is clear that $\cont{w} = \cont{w'}$. By definition of \at, it follows that $\alpha\inv(e)$ is not \at-separable from $\alpha\inv(ese)$. Thus, $(e,ese)$ is an \at-pair for $\alpha$, which yields $s ' = ese \in M_e$, as desired.

  Conversely, let $s' \in M_e$. By definition, there exists an \at-pair $(e,s) \in M^2$ with $e\in E(M)$ such that $s' = ese$.  Therefore, by definition of \at, there exist $u,v \in A^*$ such that $\cont{u} = \cont{v}$, $\alpha(u) = e$ and $\alpha(v) = s$. Let $a_1\dots,a_n \in A$ be the letters such that $v = a_1 \cdots a_n$. Since $\cont{u} = \cont{v}$, it is immediate that for each $i \leq n$, there are $x_i,y_i \in A^*$ such that $u = x_ia_iy_i$. Hence $e = \alpha(u) \Jord \alpha(a_i)$ and we conclude that $s = \alpha(a_1 \cdots a_n) \in N_e$. Consequently, $s' = ese \in eN_ee$, as desired.
\end{proof}

\subsection{Proof of Theorem~\ref{thm:utlmain}}

We fix a \vari~\Cs, a regular language $L \subseteq A^*$ and its syntactic morphism $\alpha: A^* \to M$ for the proof. We prove that $L \in \tlc{\Cs}$ if and only if all \Cs-orbits for $\alpha$ belong to \davar. We start with the left-to-right implication.

\medskip
\noindent
\textbf{\textsf{From \tlc{\Cs} to \davar.}} This direction follows from results of~\cite{pzupol2}. To use them, we need some preliminary terminology. We introduce equivalence relations connected to the class \tlc{\Cs} when \Cs is a \vari. Given a morphism $\eta: A^*\to N$ into a finite monoid~$N$, denote by $\Cs_\eta$ be the class of all languages recognized by $\eta$. The following fact is easy (see~\cite[Fact~9.3]{pzupol2}).

\begin{restatable}{fact}{esuit} \label{fct:esuit}
  Let \Cs be a \vari. For every \tla{\Cs} formula $\varphi$, there exists a \Cs-morphism $\eta: A^* \to N$ such that $\varphi$ is a \tla{\Cs_\eta} formula.
\end{restatable}

We use the standard notion of rank of a \tla{\Cs_\eta} formula: the \emph{rank} of $\varphi$ is defined as the length of the longest sequence of nested temporal operators within its parse tree. Formally:
\begin{itemize}
  \item Any atomic formula has rank $0$.
  \item The rank of $\neg \varphi$ is the same as the rank of $\varphi$.
  \item The rank of $\varphi \vee \psi$ and $\varphi \wedge \psi$ is the maximum between the ranks of $\varphi$ and $\psi$.
  \item For every language $L \subseteq A^*$, the rank of \finallyl{\varphi} and \finallyml{\varphi} is the rank of $\varphi$ plus $1$.
\end{itemize}

Two \tla{\Cs_\eta} formulas $\varphi$ and $\psi$ are \emph{equivalent} if they have the same semantics. That is, for every $w \in A^*$ and every position $i \in \pos{w}$, we have $w,i \models \varphi \Leftrightarrow w,i \models \psi$. The following key lemma is immediate from a simple induction on the rank of \tls formulas.

\begin{restatable}{lemma}{rank}\label{lem:rank}
  Let $\eta: A^* \to N$ be a morphism into a finite monoid and let $k \in \nat$. There are only finitely many non-equivalent \tla{\Cs_\eta} formulas of rank at most $k$.
\end{restatable}

We now define equivalence relations. Let $\eta: A^* \to N$ be a morphism into a finite monoid and let $k \in \nat$. Given $w,w' \in A^*$, $i \in \pos{w}$ and $i'\in \pos{w'}$, we write, $w,i \tleqk w',i'$~when:
\[
  \text{For every \tla{\Cs_\eta} formula $\varphi$ of rank at most $k$,} \quad w,i \models \varphi \Longleftrightarrow w',i' \models \varphi.
\]
It is straightforward that \tleqk is an equivalence relation. Moreover, it is immediate from the definition and Lemma~\ref{lem:rank}, that \tleqk has finite index. We lift each relation \tleqk to~$A^*$ (abusing terminology, we also denote by \tleqk the new relation): given $w,w' \in A^*$, we write $w \tleqk w'$ when $w,0 \tleqk w',0$. Clearly, $\tleqk$ is an equivalence relation of finite index over~$A^*$. Moreover, we have the following connection between \tlc\Cs and the relations \tleqk.

\begin{restatable}{lemma}{operequ}\label{lem:operequ}
  Let \Cs be a \vari and $L \subseteq A^*$. If $L \in \tlc{\Cs}$, then there exists a \Cs-morphism $\eta: A^* \to N$ and $k \in \nat$ such that $L$ is a union of \tleqk-classes.
\end{restatable}

\begin{proof}
  Let $L \in \tlc{\Cs}$. There exists a \tla{\Cs} formula $\varphi$ such that $w \in L \Leftrightarrow w,0 \models \varphi$ for all $w\in A^*$. By Fact~\ref{fct:esuit}, there exists a \Cs-morphism $\eta: A^* \to N$ such that $\varphi$ is a \tla{\Cs_\eta} formula. Let $k\in\nat$ be the rank of $\varphi$. We prove that $L$ is a union of \tleqk-classes. Given $w,w' \in A^*$ such that $w \tleqk w'$, we have to prove that $w \in L \Leftrightarrow w' \in L$. By symmetry, we only prove the left to right implication. Thus, we assume that $w \in L$. By definition of $\varphi$, it follows that  $w,0\models\varphi$. Moreover, since $w \tleqk w'$ (\emph{i.e.}, $w,0 \tleqk w',0$) and $\varphi$ is a \tla{\Cs_\eta} formula of rank $k$, we have $w',0 \models\varphi$ by definition of \tleqk. Hence, $w' \in L$ by definition of $\varphi$, as~desired.
\end{proof}

In addition to the link stated in Lemma~\ref{lem:operequ} between \tlc\Cs and the equivalence relations \tleqk, we use a property of \tleqk that follows from
\cite[Lemma~9.6 and Proposition~9.7]{pzupol2}.

\begin{restatable}{proposition}{efg} \label{prop:efg}
  Consider a morphism $\eta: A^* \to N$ into a finite monoid, let $f \in E(N)$ be an idempotent, let $u,v,z \in \eta\inv(f)$ and let $x,y\in A^*$. For every $k \in \nat$, we have:
  \[
    x(z^{k}uz^{2k}vz^{k})^{k}(z^{k}uz^{2k}vz^{k})^{k}y \tleqk x(z^{k}uz^{2k}vz^{k})^{k} z^kvz^k(z^{k}uz^{2k}vz^{k})^{k}y.
  \]
\end{restatable}

We are ready to conclude this direction of the proof: assuming that $L \in \tlc{\Cs}$, we show that all \Cs-orbits for its syntactic monoid belong to \davar. Let $e\in E(M)$ and $M_e$ be its \Cs-orbit. Proving that $M_e\in\davar$ amounts to proving that any elements $s,t\in M_e$ satisfy~\eqref{eq:da}. Fix $e,s,t\in E(M)\times M_e\times M_e$. Lemma~\ref{lem:operequ} yields a \Cs-morphism $\eta: A^* \to N$ and $k \in \nat$ such that $L$ is a union of \tleqk-classes. Since $\eta$ is a \Cs-morphism, Lemma~\ref{lem:orbitnec} yields $f \in E(N)$ such that $M_e \subseteq \alpha(\eta\inv(f))$. Since $e,s,t \in M_e$, we get $z,u,v \in A^*$ such that $z,u,v \in \eta\inv(f)$, $\alpha(z) = e$, $\alpha(u) = s$ and $\alpha(v) = t$. Let $x,y\in A^*$ be two arbitrary words. By Proposition~\ref{prop:efg}, we obtain,
\[
  x(z^{k}uz^{2k}vz^{k})^{k}(z^{k}uz^{2k}vz^{k})^{k}y \tleqk x(z^{k}uz^{2k}vz^{k})^{k} z^kvz^k(z^{k}uz^{2k}vz^{k})^{k}y.
\]
Since $L$ is a union of \tleqk-classes, the words $(z^kuz^kvz^k)^{2k}$ and $(z^kuz^kvz^k)^{k} z^kvz^k(z^kuz^kvz^k)^{k}$ are equivalent for the syntactic congruence of $L$, so they have the same image under its syntactic morphism~$\alpha$. Since $e \in E(M)$, this yields $(esete)^{2k} =(esete)^{k} ete(esete)^{k}$. Hence, $(st)^{2k} =(st)^{k} t(st)^{k}$ since $e,s,t \in M_e$ and $e$ is neutral in $M_e$ by Lemma~\ref{lem:orbitmono}. It now suffices to multiply by enough copies of $st$ on both sides to get $(st)^\omega = (st)^\omega t (st)^\omega$. Therefore, \eqref{eq:da}~holds.

\newcommand{\cdelta}{\ensuremath{\Cs_\delta}\xspace}
\medskip
\noindent
\textbf{\textsf{From \davar to \tlc{\Cs}.}} Assuming that every \Cs-orbit for the syntactic morphism $\alpha:A^*\to M$ of~$L$ belongs to \davar, we have to show that $L \in \tlc{\Cs}$, \emph{i.e.}, to build a \tla\Cs formula defining~$L$. Let us start by giving a high-level overview of the proof for this direction.

Since \tlc{\Cs} is closed under union, it suffices to prove that for all $s\in M$, the language $\alpha\inv(s)$ is in $\tlc{\Cs}$. We achieve this by inductively constructing a \tla{\Cs} formula defining $\alpha\inv(s)$. According to Lemma~\ref{lem:pairmor}, there exists a \Cs-morphism $\eta: A^* \to N$ such that the \Cs-pairs for $\alpha$ are exactly the $\eta$-pairs for $\alpha$. We use $\eta$ to leverage the assumption that all \Cs-orbits for $\alpha$ belong to \davar. More precisely, $\eta$ recognizes all the basic languages in \Cs that we shall use in our \tla{\Cs} formulas. The induction proceeds as follows: using~$\eta$, we define a sequence of languages $K_0\supseteq K_1 \supseteq \cdots \supseteq K_{|N|}$ and show by induction on $|N| - \ell$ that $K_\ell \cap \alpha\inv(s)$ can be defined by a \tla{\Cs} formula for each $\ell\leq |N|$. The induction basis is the case $\ell=|N|$, which is simple because $K_{|N|}$ is a finite language. Furthermore, the case $\ell=0$ gives the desired result since $K_0$ contains all words. The induction step consists in building a \tla{\Cs} formula describing $K_\ell\cap\alpha\inv(s)$ from several \tla{\Cs} formulas that describe the languages $K_{\ell+1} \cap \alpha\inv(t)$ for all $t \in M$. However, the actual argument is slightly more involved. Indeed, in order to perform the induction step, we must abstract each word in $K_\ell\cap\alpha\inv(s)$ by considering a specific decomposition of this word and viewing each infix as a new letter. We then argue that the resulting word belongs to $K_{\ell+1}\cap\alpha\inv(s)$, which allows us to apply induction. Yet, for this process to work, the letter that we use to abstract an infix must have the same images as this original infix under both $\alpha$ and $\eta$. This is problematic, because such a letter does not necessarily exist. We solve this issue by considering an extended alphabet~$B$, replacing $\alpha:A^*\to M$ and $\eta:A^*\to N$ with two new morphisms $\beta:B^*\to M$ and $\delta:B^*\to N$ that have the required property. Of course, this involves some preliminary work: we must reformulate both our objective (proving that all languages $\alpha\inv(s)$ can be defined in \tlc{\Cs}) and our hypothesis (that every \Cs-orbit for $\alpha$ belongs to \davar) on the new morphisms $\beta$ and $\eta$.

\medskip
We now start the proof by first defining $\beta$ and $\delta$. Recall that $\eta: A^* \to N$ is the \Cs-morphism provided by Lemma~\ref{lem:pairmor}: it is such that the \Cs-pairs for $\alpha$ are exactly the $\eta$-pairs for~$\alpha$. We fix $\eta$ for the entire proof. We define an auxiliary alphabet $B$. Let $P \subseteq M \times N$ be the set of all pairs $(\alpha(w),\eta(w)) \in M\times N$ where $w \in A^+$ is a nonempty word. For each pair $(s,r) \in P$, we create a fresh letter $b_{s,r} \not\in A$ and we define $B = \{b_{s,r} \mid (s,r) \in P\}$.

Let $\beta: B^* \to M$ and $\delta: B^* \to N$ be the morphisms defined by $\beta(b_{s,r}) = s$ and $\delta(b_{s,r}) = r$ for $(s,r) \in P$. By definition, we have $(\beta(w),\delta(w)) \in P$ for all $w \in B^+$. Let \cdelta be the class of all languages over $B$ recognized by $\delta$. One can check that \cdelta is a \vari. We now reduce membership of inverse images under $\alpha$ to $\tlc\Cs$ to that of inverse images under $\beta$ to~\tlc\cdelta.

\begin{restatable}{lemma}{thebeta} \label{lem:beta}
  For every $F \subseteq M$, if $\beta\inv(F) \in \tlc{\cdelta}$, then $\alpha\inv(F) \in \tlc{\Cs}$.
\end{restatable}

\begin{proof}
  We first define a morphism $\gamma: A^* \to B^*$. Consider a letter $a \in A$. By definition, $(\alpha(a),\eta(a)) \in P$. Hence, we may define $\gamma(a) = b_{\alpha(a),\eta(a)} \in B$. By definition, we have $\alpha(w) = \beta(\gamma(w)) \in M$ and $\eta(w) = \delta(\gamma(w))$ for every $w \in A^*$. It follows that for every $F \subseteq M$, we have $\alpha\inv(F) = \gamma\inv(\beta\inv(F)) \subseteq A^*$. Consequently, it now suffices to prove that for every $K \subseteq B^*$ such that $K \in \tlc{\cdelta}$, we have $\gamma\inv(K) \in \tlc{\Cs}$. We fix $K$ for the proof. Since $K \in \tlc{\cdelta}$, it is defined by a formula $\psi \in \tla{\cdelta}$. We apply two kinds of modifications to $\psi$ in order to build a new formula $\psi' \in \tla{\Cs}$ defining $\gamma\inv(K)$:
  \begin{enumerate}
    \item We replace every atomic subformula ``$b$'' for $b \in B$ by the \tla{\Cs}-formula $\bigvee_{\{a \in A \mid \gamma(a) = b\}} a$.
    \item For every temporal modality $\textup{F}_H$ (resp. $\textup{P}_H$) occurring in $\psi$, we have $H \in \cdelta$ by hypothesis. Hence, $H$ is recognized by $\delta$ and there exists $G \subseteq N$ such that $H = \delta\inv(G)$. Note that $\eta\inv(G) \in \Cs$ since $\eta$ is a \Cs-morphism. We replace the temporal modality $\textup{F}_H$ (resp. $\textup{P}_H$) by  $\textup{F}_{\eta\inv(G)}$ (resp. $\textup{P}_{\eta\inv(G)}$).
  \end{enumerate}
  By definition the resulting formula $\psi'$ belongs to \tla{\Cs} and one can verify that for every $w \in A^*$, we have $w,0 \models \psi' \Leftrightarrow \gamma(w),0 \models \psi$. Since $L_{min}(\psi) = K$, we get $L_{min}(\psi') = \gamma\inv(K)$, which implies that $K \in \tlc{\Cs}$. This completes the proof.
\end{proof}

In view of Lemma~\ref{lem:beta}, it suffices to prove that any language recognized by $\beta$ belongs to \tlc{\cdelta}. Since $L$ is recognized by $\alpha$, this will imply $L \in \tlc{\Cs}$, which is our goal. In the next lemma, we reformulate on $\beta$ and $\delta$ the assumption that every \Cs-orbit for $\alpha$ belongs to~\davar.

\begin{restatable}{lemma}{hypbeta} \label{lem:hypbeta}
  For every $e \in E(M)$ and every $s,t \in M$, if $(e,s)$ and $(e,t)$ are $\delta$-pairs for $\beta$, then $(esete)^{\omega} = (esete)^{\omega}ete(esete)^{\omega}$.
\end{restatable}

\begin{proof}
  By hypothesis, there exist $u,v,x,y \in B^*$ such that $\delta(u) = \delta(v)$, $\delta(x) = \delta(y)$, $\beta(u) = \beta(x) = e$, $\beta(v) = s$ and $\beta(y) = t$. The definitions of $\beta$ and $\delta$ imply that for any $w \in B^*$, there exists $w' \in A^*$ such that $\delta(w) = \eta(w')$ and $\beta(w) = \alpha(w')$. Therefore, we obtain $u',v',x',y' \in A^*$ such that $\eta(u') = \eta(v')$, $\eta(x') = \eta(y')$, $\alpha(u') = \alpha(x') = e$, $\alpha(v') = s$ and $\alpha(y') = t$. Thus, $(e,s) \in M^2$ and $(e,t) \in M^2$ are $\eta$-pairs for $\alpha$. By definition of $\eta$, it follows that they are \Cs-pairs for $\alpha$. Hence, $ese$ and $ete$ both belong to the \Cs-orbit of $e$ for $\alpha$. Since all \Cs-orbits for $\alpha$ belong to \davar by hypothesis, this gives $(esete)^{\omega} = (esete)^{\omega}ete(esete)^{\omega}$.
\end{proof}

\newcommand{\dpj}[1]{\ensuremath{d_{\Jrel}(#1)}\xspace}

We now use the Green relation \Jrel over $N$ to associate a number $\dpj{r} \in \nat$ with every element $r \in N$. We let $\dpj{r}$ be the maximal number $n\in\nat$ such that there exist $n$ elements $r_1,\dots,r_n \in N$ satisfying $r \Jords r_1 \Jords \cdots \Jords r_n$. By definition, $0 \leq \dpj{r} \leq |N| - 1$. In particular, we have $\dpj{r} = 0$ if and only if $r$ is maximal for \Jord (\emph{i.e.}, if and only if $r \Jrel 1_N$). Finally, given a word $w \in B^*$, we write $\dpj{w} \in \nat$ for $\dpj{\delta(w)}$. Observe that for all $x,y,z \in B^*$, we have $\dpj{y} \leq \dpj{xyz}$ (as $xyz\Jord y$), a fact that we shall use frequently.

\medskip\noindent
In order to argue inductively, we define a family of languages $K_\ell \subseteq B^*$ for $\ell \in \nat$ as follows:
\[
  K_\ell = \bigl\{w \in B^* \mid \text{for all $k \leq \ell$ and $x,y,z \in B^*$, if $w = xyz$ and $|y| = k$, then $\dpj{y} \geq k$}\bigr\}.
\]
Note that $K_0 = B^*$ as $\dpj{y} \geq 0$ for all $y \in B^*$. Also, if $\ell \geq |N|$, then $K_{\ell}$ is \emph{finite} (it contains words of length at most $|N|-1$ as $\dpj{y} < |N|$ for all $y \in B^*$). We now have the next~lemma.

\begin{restatable}{lemma}{llpo} \label{lem:llpo}
  Let $\ell \in \nat$ and $w \in K_\ell$. Then $\dpj{w} \leq \ell$ if and only if for all $x,y,z \in B^*$ such that $w = xyz$ and $|y| \leq \ell+1$, we have $\dpj{y} \leq \ell$.
\end{restatable}

\begin{proof}
  The ``only if'' direction is immediate since $\dpj{y} \leq \dpj{w}$ for every infix $y$ of $w$. Conversely, assume that for all $x,y,z \in B^*$ such that $w = xyz$ and $|y| \leq \ell+1$, we have $\dpj{y} \leq \ell$. We prove that $\dpj{w} \leq \ell$. If $|w|\leq\ell + 1$, this is immediate. Assume now that $|w| > \ell+1$. We get $b_1,\dots,b_n \in B$ and $v \in B^*$ such that $|v| = \ell+1$ and $w =vb_1 \cdots b_n$. We use induction on $i$ to prove that $\delta(v) \Jrel \delta(vb_1 \cdots b_i)$ for all $i \leq n$. Since $\dpj{v} \leq \ell$ by hypothesis, the case $i = n$ yields $\dpj{w} \leq \ell$. The case $i = 0$ is trivial: we have $\delta(v) \Jrel \delta(v)$. Assume now that $i \geq 1$. By induction hypothesis, we know that $\delta(v) \Jrel \delta(vb_1 \cdots b_{i-1})$. Let $x,y \in B^*$ such that $|y| = \ell$ and $xy = vb_1 \cdots b_{i-1}$ (the words $x$ and $y$ exist because $|v| = \ell+1$). Since $w \in K_\ell$, and $y$ is an infix of $w$ such that $|y| = \ell$, we know that $\dpj{y} \geq \ell$. Moreover, $yb_i$ is an infix of $w$ such that $|yb_i| = \ell+1$, which yields $\dpj{yb_i} \leq \ell$ by hypothesis. Since $\dpj{y} \leq \dpj{yb_i}$, we get $\dpj{yb_i} = \dpj{y} = \ell$, which implies that $\delta(yb_i) \Jrel \delta(y)$. Moreover, we have $\delta(yb_i) \Rord \delta(y)$. Thus, Lemma~\ref{lem:jlr} yields $\delta(yb_i) \Rrel \delta(y)$. This implies that $\delta(xyb_i) \Rrel \delta(xy)$. Hence, $\delta(vb_1 \cdots b_{i}) \Jrel \delta(vb_1 \cdots b_{i-1}) \Jrel \delta(v)$. This completes the proof.
\end{proof}

We now prove that for all $s \in M$ and $\ell \in\nat$, we have $K_{\ell}\cap\beta\inv(s) \in \tlc{\cdelta}$. Our objective (every language recognized by $\beta$ belongs to \tlc{\cdelta}) follows from the case $\ell = 0$, since $K_0=B^*$. The proof involves two steps. The first settles the case of elements of~$K_{\ell}\cap\beta\inv(s)$ whose image under $\delta$ has a $d_{\Jrel}$ value at most~$\ell$. We do not use induction for this case, which relies on the inclusion $\upol{\bpol{\cdelta}}\subseteq \tlc{\cdelta}$. It is also the place where we use Lemma~\ref{lem:hypbeta}, \emph{i.e.}, the hypothesis that all \Cs-orbits \mbox{for $\alpha$ are in~\davar.}

\begin{restatable}{proposition}{bcase} \label{prop:bcase}
  Let  $(\ell,s,r) \in \nat\times M\times N\!$. If $\dpj{r} \leq \ell$ then \mbox{$K_{\ell} \cap \beta\inv(s) \cap \delta\inv(r) \in \tlc{\cdelta}$.}
\end{restatable}

\begin{proof}
  We prove that $K_{\ell} \cap \beta\inv(s) \cap \delta\inv(r) \in \upol{\bpol{\cdelta}}$, which, by Proposition~\ref{prop:ubptl}, will give the desired result $K_{\ell} \cap \beta\inv(s) \cap \delta\inv(r)\in\tlc{\cdelta}$.  Let $\gamma: B^*\to Q$ be the syntactic morphism of $K_{\ell} \cap \beta\inv(s)\cap\delta\inv(r)$. By Theorem~\ref{thm:ubpcar}, it suffices to show that given $q_1,q_2\in Q$ and $f \in E(Q)$ such that $(f,q_1) \in Q^2$ is a \cdelta-pair for $\gamma$, the following equation holds:
  \begin{equation} \label{eq:basec}
    (fq_1fq_2f)^{\omega+1} = (fq_1fq_2f)^{\omega} fq_2f (fq_1fq_2f)^{\omega}.
  \end{equation}
  Let $q_1,q_2,f\in Q$ be such elements. By definition of \cdelta, we know that $\delta$ is a $\cdelta$-morphism. Therefore, Lemma~\ref{lem:pairmor} implies that $(f,q_1)$ is a $\delta$-pair for~$\gamma$. We get $u',v'_1 \in B^*$, such that $\delta(u') = \delta(v'_1)$, $\gamma(u')=f$ and $\gamma(v'_1)=q_1$. Note that if $v'_1 = \veps$, then $q_1 = 1_Q$ and~\eqref{eq:basec}~holds since it is clear that $(fq_2f)^{\omega+1}=(fq_2f)^{2\omega+1}$. Therefore, we assume from now on that $v'_1 \in B^+$. Let us also choose $v'_2 \in B^*$ such that $\gamma(v'_2) = q_2$. We now define $p = \ell \times \omega(N) \times \omega(M) \times \omega(Q)$, $u = (u')^p$, $v_1 = (u')^{p-1}v'_1$ and $v_2 = uv'_2u(uv_1uv'_2u)^{p-1}$. We compute $\gamma(u) = f$, $\gamma(v_1) = fq_1$ and $\delta(u) = \delta(v_1)$. Moreover, since $p$ is a multiple of $\omega(N)$, the element $\delta(u) = \delta(v_1)$ is an idempotent $g \in E(N)$. Finally, we have $\gamma(v_2) = fq_2f(fq_1fq_2f)^{p-1}$ and $\delta(v_2) = (g\delta(v'_2)g)^p$. In particular, it follows that $\delta(v_2)$ is an idempotent $h \in E(N)$ such that $gh = hg = h$.

  We prove that $(uv_1uv_2u)^{p}$ and  $(uv_1uv_2u)^{p} uv_2u (uv_1uv_2u)^{p}$ are equivalent for the syntactic congruence of $K_{\ell} \cap \beta\inv(s)\cap \delta\inv(r)$. This will imply that they have the same image under~$\gamma$, which yields $(fq_1fq_2f)^{\omega}  = (fq_1fq_2f)^{\omega}  fq_2f(fq_1fq_2f)^{2\omega-1}$. One may then multiply by $fq_1fq_2f$ on the right to get~\eqref{eq:basec}, as desired. For $x,y \in A^*$, let $z_1 = x (uv_1uv_2u)^{p} y$ and $z_2 = x(uv_1uv_2u)^{p} uv_2u (uv_1uv_2u)^{p}y$. We have to show that $z_1 \in K_{\ell} \cap \beta\inv(s)\cap\delta\inv(r)$ if and only if $z_2 \in K_{\ell}  \cap \beta\inv(s)\cap\delta\inv(r)$. We first treat the special case where $|u|<\ell$.

  Assume that $|u| < \ell$. We show that in this case $z_1 \not\in K_\ell$ and $z_2 \not\in K_\ell$ (which implies the desired result). Since $u = (u')^p$ and $p\geq\ell$, the hypothesis that $|u| < \ell$ yields $u = u' = \veps$. Since $\delta(u)= \delta(v_1)$, we get $\delta(v_1) = 1_N$. Recall that $v_1 = (u')^{p-1}v'_1$ and $v'_1 \in B^+$ by hypothesis. Thus, $v_1\in B^+$, which means that it contains a letter $b \in B$ such that $\delta(b) \Jrel 1_N$. In particular $\dpj{b} = 0$. Hence, $b$ is an infix of length $1$ of both $z_1$ and $z_2$ such that $\dpj{b} < 1$. Now $\ell >|u|=0$, so that $\ell \geq 1$. This implies $z_1 \not\in K_\ell$ and $z_2 \not\in K_\ell$.
  This completes the special~case.

  From now on, we assume that $|u| \geq \ell$. Since $\delta(u) = \delta(v_1) = g \in E(N)$, $\delta(v_2)=h\in E(N)$ and $gh=hg=h$, we have $\delta(z_1) = \delta(z_2) = \delta(x)h\delta(y)$.  Therefore, $z_1 \in  \delta\inv(r)$ if and only if $z_2 \in \delta\inv(r)$. Let us prove that $z_1 \in K_{\ell} \Leftrightarrow z_2 \in K_{\ell}$. This is trivial if $\ell = 0$ since $K_0 = B^*$. Assume now that $\ell \geq 1$. Since $|u| \geq \ell$ by hypothesis, it follows that for every $k \leq \ell$, $z_1$ and $z_2$ have the same infixes of length $k$. This implies that $z_1 \in K_{\ell} \Leftrightarrow z_2 \in K_{\ell}$, as~desired.

  It remains to prove that if $z_1,z_2 \in K_\ell \cap \delta\inv(r)$, then $\beta(z_1) = \beta(z_2)$. We first show that our assumptions imply $g \Jrel h$. Again, there are two cases. First, assume that $\ell = 0$. Since $\dpj{r} \leq \ell$ by hypothesis, we get $r \Jrel 1_N$. Thus, since $u$ and $v_2$ are infixes of $z_1 \in \delta\inv(r)$, we have $\delta(u) \Jrel \delta(v_2) \Jrel 1_N$, which exactly says that $g \Jrel h \Jrel 1_N$. Assume now that $\ell \geq 1$. Recall that $|u| \geq \ell$. Since $u$ is an infix of $v_2$, this also implies that $|v_2| \geq \ell$. Hence, since $u$ and $v_2$ are infixes of $z_2 \in K_\ell \cap \delta\inv(r)$, we get $\dpj{u} \geq \ell$ and $\dpj{v_2} \geq \ell$, $r \Jord \delta(u)$ and $r \Jord \delta(v_2)$. In particular, it follows that $\dpj{r} \geq \dpj{u} \geq \ell$ and $\dpj{r} \geq \dpj{v_2} \geq \ell$. Since $\dpj{r} \leq \ell$ by hypothesis on $r$, we get $\dpj{r} = \dpj{u} = \dpj{v_2} =\ell$. Together with $r \Jord \delta(u)$ and $r \Jord \delta(v_2)$, this yields $r \Jrel \delta(u) \Jrel \delta(v_2)$, \emph{i.e.}, $r \Jrel g \Jrel h$. This completes the proof that $g \Jrel h$. Since we also know that $hg = gh = h$, we have $h \Rord g$ and Lemma~\ref{lem:jlr} yields $g \Rrel h$. We get $z \in N$ such that $g = hz$. Thus, we have $h = hg = hhz = hz = g$.

  Altogether, we obtain $\delta(u) = \delta(v_1) = \delta(v_2) = g \in E(N)$. This implies that $(\beta(u),\beta(v_1))$ and $(\beta(u),\beta(v_2))$ are $\delta$-pairs for $\beta$. Moreover, recall that $u = (u')^p$ where $p$ is a multiple of $\omega(M)$. Hence, we have $\beta(u) \in E(M)$. Consequently, it follows from Lemma~\ref{lem:hypbeta} that $\beta((uv_1uv_2u)^p) = \beta((uv_1uv_2u)^p uv_2u(uv_1uv_2u)^p)$. It now suffices to multiply by $\beta(x)$ on the left and $\beta(y)$ on the right to obtain $\beta(z_1) = \beta(z_2)$, as desired.
\end{proof}

We now turn to the second step of the proof, which is formalized in the following statement.

\begin{restatable}{proposition}{pmain} \label{prop:pmain}
  Let $\ell \leq |N|$ and $s \in M$. There exists a \tla{\cdelta} formula $\varphi_{\ell,s}$ such that for every $w \in K_\ell$, we have $w,0 \models \varphi_{\ell,s} \Leftrightarrow \beta(w) = s$.
\end{restatable}

Let us first use Proposition~\ref{prop:pmain} to complete the main proof: we have to show that every language recognized by $\beta$ belongs to \tlc{\cdelta}. Clearly, it suffices to show that $\beta\inv(s) \in \tlc{\cdelta}$ for each $s \in M$. We apply Proposition~\ref{prop:pmain} for $\ell = 0$. Since $K_0=B^*$, this yields a formula $\varphi_{0,s} \in \tla{\cdelta}$ such that $L_{min}(\varphi_{0,s}) = \beta\inv(s)$. Thus, $\beta\inv(s) \in \tlc{\cdelta}$, as desired.

\medskip

It remains to prove Proposition~\ref{prop:pmain}. We construct $\varphi_{\ell,s} \in \tla{\cdelta}$ by induction on $|N| - \ell$. If $\ell = |N|$, we define $\varphi_{\ell,s}$ so that $L_{min}(\varphi_{\ell,s})=K_{\ell}\cap\beta\inv(s)$. Since $K_{|N|}\cap\beta\inv(s)$ is finite and $\tla{\cdelta}$ is closed under disjunction, it suffices to build for every word $w\in B^*$ a \tla{\cdelta} formula $\varphi_{w}$ defining~$\{w\}$. Since $B^{*}\in\cdelta$, one may use the ``$\finally{}\!$'' modality. For $w=b_1\cdots b_n$,~let
\[
  \psi_w=\finally(b_1 \wedge \finally{(b_2 \wedge \finally{(b_3 \wedge \cdots \wedge\finally{b_n})})}).
\]
One may then choose $\varphi_w=\psi_w\wedge\bigwedge_{\mathit{u}\in B^{*},|u|=|w|+1}\neg\psi_u$.

\smallskip
Assume now that $\ell < |N|$. We present a construction for splitting the words in $K_\ell$ into two parts: a prefix mapped to an element $r \in N$ such that $\dpj{r} \leq \ell$ (we handle it with Proposition~\ref{prop:bcase}) and a suffix that we abstract as a word in $K_{\ell+1}$ \mbox{(we handle it by~induction).}

Let $w \in K_\ell$. For each position $i \in \pos{w} \setminus \{0\}$ and $k \in \nat$, we write $\sigma_k(w,i) \in B^*$ for the infix $\infix{w}{i-1}{j}$ where $j = \omin(i+k,|w|+1)$. In other words, $\sigma_k(w,i)=\wpos{w}{i} \cdots \wpos{w}{i+k-1}$ if $i+k-1\leq |w|$ and $\sigma_k(w,i) = \wpos{w}{i} \cdots \wpos{w}{|w|}$ otherwise. In particular, we have $|\sigma_k(w,i)| \leq k$.

\begin{restatable}{lemma}{lem:checkdist}\label{lem:checkdist}
  Let $k \leq \ell+1$ and $u \in B^*$ be such that $|u| \leq k$. There exists a formula $\pi_{k,u} \in \tla{\cdelta}$ such that for all $w \in K_\ell$ and $i \in \pos{w} \setminus \{0\}$, $w,i \models \pi_{k,u} \Leftrightarrow\sigma_{k}(w,i) = u$.
\end{restatable}

\begin{proof}
  If $u = \veps$, it suffices to define $\pi_{k,u} = \top$ when $k = 0$ and $\pi_{k,u} = max$ when $k \geq 1$. Assume now that $|u| \geq 1$. If $k \leq 1$, it follows that $|u| = 1=k$. Hence, $u$ is a letter $b \in B$ and it suffices to define $\pi_{k,u} = b$. Assume now that $k \geq 2$. Let $C\subseteq B$ be the set of letters mapped to~$1_N$ under $\delta$, so that $H\stackrel{\smash{\text{\tiny def}}}=\delta\inv(1_N)=C^*$. By definition, $H\in\Cs_{\delta}$.
  Since $2\leq k \leq \ell+1$, we have $\ell \geq 1$, which implies, by definition of $K_\ell$, that no word of $K_\ell$ can contain a letter $b$ with $\dpj b=0$. In particular, words of $K_{\ell}$ cannot contain letters of $C$. Therefore, if $w\in K_{\ell}$, $i\in\pos{w}$ and $\psi\in\tla\cdelta$, we have $w,i\models \finallyp{H}\psi$ if and only if $w,i+1\models \psi$. Let $u = b_1 \cdots b_n$ (with $b_i\in B$). We have $n = |u| \leq k$ by hypothesis. We consider two cases for defining $\pi_{k,u}$:
  \begin{itemize}
    \item If $n = k$, we let $\pi_{k,u} = (b_1 \wedge \finallyp{H}{(b_2 \wedge \finallyp{H}{(b_3 \wedge \cdots \finallyp{H}{b_n})})})$.
    \item If $n < k$, we let $\pi_{k,u} = (b_1 \wedge \finallyp{H}{(b_2 \wedge \finallyp{H}{(b_3 \wedge \cdots \finallyp{H}{(b_n \wedge \finallyp{H}{max})})})})$.
  \end{itemize}
  The above fact on $\finallyp{H}{}$ implies that this definition fulfills the desired property.
\end{proof}

\noindent
\textbf{\textsf{Pointed positions.}} Consider $w \in K_\ell$. We say that an arbitrary position $i \in \pos{w}$ is  \emph{pointed} when either $i \in \{0,|w|+1\}$, or $i \in \posc{w}$ and $\dpj{\sigma_{\ell+1}(w,i)} \geq \ell+1$.
\begin{definition}[Detection of pointed positions in \tla{\cdelta}]\label{rem:pointed}
  Let $\pi = min \vee max \vee \bigvee_{u \in U} \pi_{\ell+1,u}$ where $U = \{u \in B^* \mid |u| \leq \ell+1 \text{ and } \dpj{u} \geq \ell+1\}$. By definition of $\pi_{\ell+1,u}$ in Lemma~\ref{lem:checkdist}, we know that for  $w \in K_\ell$ and $i \in \pos{w}$, we have $w,i \models \pi$ if and only if position $i$ is~pointed.
\end{definition}
A position $i \in \pos{w}$ which is \emph{not} pointed is said to be \emph{safe}. We now prove that we may constrain the evaluation of \tla{\cdelta} formulas to infixes that only contain \emph{safe} positions.

\begin{restatable}{lemma}{constrain} \label{lem:constrain}
  Let $\psi \in \tla{\cdelta}$ and $H \in \cdelta$. There exist two formulas $\xfinallyp{H}{\psi}$ and $\xfinallymp{H}{\psi}$ of $\tla{\cdelta}$ such that for all $w \in K_\ell$ and all $i \in \pos{w}$, the two following properties hold:
  \begin{itemize}
    \item $w,i \models \xfinallyp{H}{\psi}$ if and only if there exists $j \in \pos{w}$ such that $j > i$, $w,j \models \psi$, $\infix{w}{i}{j} \in H$ and all positions $h \in \pos{w}$ such that $i < h < j$ are safe.
    \item $w,i \models \xfinallymp{H}{\psi}$ if and only if there exists $j \in \pos{w}$ such that $j < i$, $w,j \models \psi$, $\infix{w}{j}{i} \in H$ and all positions $h \in \pos{w}$ such that $j < h < i$ are safe.
  \end{itemize}
\end{restatable}

\begin{proof}
  We begin by characterizing infixes containing only safe positions. Let $w \in K_\ell$ and $i,j \in \pos{w}$ be such that $i < j$. We prove that the following two properties are equivalent:
  \begin{enumerate}
    \item All positions $h \in \pos{w}$ such that $i < h < j$ are safe.
    \item Either $\delta(\infix{w}{i}{j}) = 1_N$ or $\dpj{\infix{w}{i}{j}\sigma_{\ell}(w,j)} \leq \ell$.
  \end{enumerate}
  Assume first that all positions $h \in \pos{w}$ such that $i < h < j$ are safe. If $i+1 = j$, then $\infix{w}{i}{j} = \veps$, whence $\delta(\infix{w}{i}{j}) = 1_N$. Assume now that $i+1 < j$. Observe that $\infix{w}{i}{j}\sigma_{\ell}(w,j)$ belongs to $K_\ell$ since it is an infix of $w \in K_\ell$. Moreover, since $i+1 < j$, there exists at least one $h \in \pos{w}$ such that $i < h < j$. Combined with the assumption that all such positions $h$ are safe, this implies that for every $x,y,z \in B^*$ such that $xyz = \infix{w}{i}{j}\sigma_{\ell}(w,j)$ and $|y| \leq \ell +1$, we have $\dpj{y} \leq \ell$. Therefore, Lemma~\ref{lem:llpo} entails that $\dpj{\infix{w}{i}{j}\sigma_{\ell}(w,j)} \leq \ell$, as desired.

  Conversely, assume that either $\delta(\infix{w}{i}{j}) = 1_N$ or $\dpj{\infix{w}{i}{j}\sigma_{\ell}(w,j)} \leq \ell$. We start with the latter case. Since $\infix{w}{i}{j}\sigma_{\ell}(w,j) \in K_\ell$, Lemma~\ref{lem:llpo} implies that for every $x,y,z \in B^*$ such that $xyz = \infix{w}{i}{j}\sigma_{\ell}(w,j)$ and $|y| \leq \ell +1$, we have $\dpj{y} \leq \ell$. In particular, it follows that every $h \in \pos{w}$ such that $i < h < j$ is safe. Assume now that $\delta(\infix{w}{i}{j}) = 1_N$. If $\ell = 0$, then $\sigma_{\ell}(w,j) = \veps$ and we are back to the previous case. Otherwise, $\ell\geq 1$ and since $w \in K_\ell$, the fact that $\delta(\infix{w}{i}{j}) = 1_N$ yields $\infix{w}{i}{j} = \veps$, which completes the argument.

  We are now ready to complete the proof of the lemma. Let $\psi \in \tla{\cdelta}$ and $H \in \cdelta$. For every $r \in N$, we let $H_r = H \cap \delta\inv(r)$ and $U_r = \{u \in B^* \mid |u| \leq \ell \text{ and } \dpj{r\delta(u)} \leq \ell\}$. Observe that $H_r\in\cdelta$. Now, in view of the preliminary result, it suffices to define,
  \[
    \xfinallyp{H}{\psi}= \finallyp{H_{1_N}}{\psi} \vee \bigvee_{r \in N} \bigvee_{u \in U_r} \finallyp{H_r}{\left(\pi_{\ell,u} \wedge \psi\right)} \ \ \text{and} \ \ \xfinallymp{H}{\psi} =\finallymp{H_{1_N}}{\psi} \vee \bigvee_{r \in N} \bigvee_{u \in U_r} \left(\pi_{\ell,u} \wedge \finallymp{H_r}{\psi}\right).
  \]
  This completes the proof.
\end{proof}

\newcommand{\wch}[1]{\ensuremath{\widehat{#1}}\xspace}

\noindent \textbf{\textsf{Pointed decomposition.}} Let $w \in K_\ell$ and let $0 = i_0 < i_1 < \dots < i_n < i_{n+1} = |w|+1$ be all the pointed positions of $w$. The \emph{pointed decomposition} of $w$ is the decomposition $w = w_0b_1w_1 \cdots b_nw_n$ where the highlighted letters $b_1,\dots,b_n \in B$ are those carried by the pointed positions $i_1,\dots,i_n$. For $0 \leq j \leq n$, we associate the word $f(w,i_j) = w_j$ to the pointed position~$i_j$. Moreover, we define a new word $\wch{w} \in B^*$ built from the suffix $b_1w_1\cdots b_nw_n$. For $1 \leq j \leq n$, let $(t_j,q_j) = (\beta(b_jw_j),\delta(b_jw_j)) \in P$. By definition of $\beta$ and $\delta$, we know that there is a letter $b_{t_j,q_j} \in B$ such that $(\beta(b_{t_j,q_j}),\delta(b_{t_j,q_j})) = (t_j,q_j)$. We let $\wch{w} = b_{t_1,q_1} \cdots b_{t_n,q_n}$. Note that by definition, $\beta(b_1w_1\cdots b_nw_n) = \beta(\wch{w})$ and $\delta(b_1w_1\cdots b_nw_n) = \delta(\wch{w})$. Finally, we define a surjective map $i \mapsto \mu(i)$ associating a position $\mu(i) \in \pos{\wch{w}}$ to each \emph{pointed} position $i \in \pos{w}$: for $0 \leq j \leq n+1$, we let $\mu(i_j) = j$. We complete this definition with a key property. For every pointed position $i \in \{0\} \cup \posc{w}$, one can compute the images of the word $f(w,i)$ under $\beta$ and $\delta$ with a \tla{\cdelta} formula. This is where we use Proposition~\ref{prop:bcase}.

\newcommand{\crochi}[1]{\ensuremath{\langle #1\rangle_{min}}\xspace}
\newcommand{\crocha}[1]{\ensuremath{\langle #1\rangle_{max}}\xspace}
\newcommand{\cfinal}[1]{\ensuremath{\lfloor #1\rfloor}\xspace}

\begin{restatable}{lemma}{checkletters} \label{lem:checkletters}
  Let $(t,r) \in M \times N$. There exists $\Gamma_{t,r} \in \tla{\cdelta}$ such that for all $w \in K_\ell$ and all pointed positions $i \in \{0\} \cup \posc{w}$, we have $w,i \models \Gamma_{t,r} \Leftrightarrow \beta(f(w,i)) = t \text{ and } \delta(f(w,i)) = r$.
\end{restatable}

\begin{proof}
  First observe that by definition, if $w \in K_\ell$ and $i \in \{0\} \cup \posc{w}$ is pointed, the infix $f(w,i)$ contains only \emph{safe} positions. Hence, for every $x,y,z \in B^*$ such that $f(w,i) = xyz$ and $|y| \leq \ell+1$, we have $\dpj{y} \leq \ell$. By Lemma~\ref{lem:llpo}, it follows that $\dpj{f(w,i)} \leq \ell$. Therefore, if $\dpj{r} > \ell$, then $\delta(f(w,i))$ cannot be equal to $r$, and it suffices to define $\Gamma_{t,r} = \bot$.

  We now assume that $\dpj{r} \leq \ell$. Proposition~\ref{prop:bcase} implies that $K_\ell \cap \beta\inv(t) \cap \delta\inv(r) \in \tlc{\cdelta}$. We get a formula $\psi \in \tla{\cdelta}$ such that for every $u \in K_\ell$, we have $u,0 \models \psi$ if and only if $\beta(u)=t$ and $\delta(u) = r$. Using Lemma~\ref{lem:constrain}, we modify $\psi$ so that given $w \in K_\ell$, the evaluation of $\psi$ at a pointed position $i$ is constrained to the infix $f(w,i)$. More precisely, we use structural induction to build two formulas $\crochi{\psi}$ and $\crocha{\psi}$ such that given $w \in K_\ell$, a pointed position $i \in \{0\} \cup \posc{w}$ and $j \in \pos{f(w,i)}$, the two following properties hold:
  \begin{itemize}
    \item If $j \leq |f(w,i)|$, then $w,i+j \models \crochi{\psi} \Leftrightarrow f(w,i),j \models \psi$.
    \item If $1 \leq j$,\hspace*{6.85ex} then $w,i+j \models \crocha{\psi} \Leftrightarrow f(w,i),j \models \psi$.
  \end{itemize}
  It will then suffice to define $\Gamma_{t,r} = \crochi{\psi}$. We only describe the construction, and leave it to the reader to check that it satisfies the above properties. Note that we use the formula $\pi\in\tla\cdelta$ of Definition~\ref{rem:pointed} that detects pointed positions.

  For $\psi \in B \cup \{\top,\bot\}$, we let $\crochi{\psi} = \crocha{\psi} = \psi$. If $\psi = min$, we let $\crochi{\psi} = \pi$ and $\crocha{\psi} = \bot$. If $\psi = max$, we let $\crochi{\psi} = \bot$ and $\crocha{\psi} = \pi$. We handle Boolean operators in the expected way. For instance, we define $\crochi{\psi'\vee\psi''}=\crochi{\psi'}\vee\crochi{\psi''}$, $\crochi{\psi'\wedge\psi''}=\crochi{\psi'}\wedge\crochi{\psi''}$ and $\crochi{\neg\psi'}=\neg\crochi{\psi'}$, and similarly for $\crocha{\cdot}$. If~$\psi = \finallyp{H}{\psi'}$ for $H \in \cdelta$, we let $\crochi{\psi} = \xfinallyp{H}{\crocha{\psi'}}$ and $\crocha{\psi} = \neg \pi \wedge \xfinallyp{H}{\crocha{\psi'}}$. Symmetrically, if $\psi = \finallymp{H}{\psi'}$ for some $H \in \cdelta$, we define $\crochi{\psi} = \neg \pi \wedge \xfinallymp{H}{\crochi{\psi'}}$ and $\crocha{\psi} =\xfinallymp{H}{\crochi{\psi'}}$. This concludes the inductive construction of $\crochi{\psi}$ and $\crocha{\psi}$ and the proof of the proposition.
\end{proof}

\noindent
\textbf{\textsf{Construction of the formulas \boldmath{$\varphi_{\ell,s}$}.}} We are ready to complete the proof of Proposition~\ref{prop:pmain}.  For every $s \in M$, we build a formula $\zeta_s \in \tla{\cdelta}$ such that for every $w \in K_\ell$, we have $w,0 \models \zeta_s \Leftrightarrow \beta(\wch{w}) = s$. Given $s \in M$, it will then suffice to define $\varphi_{\ell,s} \in \tla{\cdelta}$ as follows:
\[
  \varphi_{\ell,s} = \bigvee_{\{(s_1,s_2) \in M^2 \mid s_1s_2=s\}} \Bigl(\Bigl(\bigvee_{r \in N}\Gamma_{s_1,r}\Bigr) \wedge \zeta_{s_2}\Bigr).
\]
Indeed, it is straightforward that for every word $w \in K_\ell$, we have $\beta(w) =  \beta(f(w,0)) \beta(\wch{w})$. Consequently, by definition of $\varphi_{\ell,s}$, we get $w,0 \models \varphi_{\ell,s} \Leftrightarrow \beta(f(w,0)) \beta(\wch{w}) = s \Leftrightarrow \beta(w) = s$ for all $w \in K_\ell$, which concludes the proof of Proposition~\ref{prop:pmain}. We now concentrate on building~$\zeta_s$. This is where we use induction in Proposition~\ref{prop:pmain}. Indeed, we have the following lemma.

\begin{restatable}{lemma}{basprops1} \label{lem:bprops1}
  For every $w \in K_\ell$, we have $\wch{w} \in K_{\ell+1}$.
\end{restatable}

\begin{proof}
  Let $k \leq \ell+1$ and $x,y,z \in B^*$ such that $\wch{w} = xyz$ and $|y| = k$. We have to prove that $\dpj{y} \geq k$. Let $w = w_0b_1w_1 \cdots b_nw_n$ be the pointed decomposition of~$w$. By definition of $\wch{w}$, we have $\delta(y) = \delta(b_{h}w_{h} \cdots b_{h+k-1}w_{h+k-1})$ for some $h\leq n$. Let $u=b_{h}w_{h} \cdots b_{h+k-1}w_{h+k-1}$. We have to show that $\dpj{y}=\dpj{u}\geq k$. Clearly, $|u| \geq k$. Hence, if~$k \leq \ell$, the hypothesis that $w \in K_\ell$ yields $\dpj{u} \geq k$. Otherwise, $k = \ell+1$. Thus, $|u| \geq \ell+1$ and since the position labeled by $b_h$ in $w$ is pointed, this yields $\dpj{u} \geq \ell+1$. In both cases, we get~$\dpj{y} \geq k$.
\end{proof}

Let $s \in M$. In view of Lemma~\ref{lem:bprops1}, induction on $|N|-\ell$ in Proposition~\ref{prop:pmain} yields a \tla{\cdelta} formula $\psi_s$ such that for every $w \in K_\ell$, we have $\wch{w},0 \models \psi_s \Leftrightarrow \beta(\wch{w}) = s$. Thus, it now suffices to prove that for every $\psi \in \tla{\cdelta}$, there exists a formula $\cfinal{\psi} \in \tla{\cdelta}$ such that for every $w \in K_\ell$ and every pointed position $i \in \pos{w}$, we have $w,i \models \cfinal{\psi} \Leftrightarrow \wch{w},\mu(i) \models \psi$. It will then follow, for $i=0$, that $w,0 \models \cfinal{\psi_s} \Leftrightarrow  \beta(\wch{w}) = s$, meaning that we can define~$\zeta_s = \cfinal{\psi_s}$.

We construct \cfinal{\psi} by structural induction on $\psi$. If $\psi \in \{min,max,\top,\bot\}$, we let $\cfinal{\psi} = \psi$. Suppose now that $\psi = b_{t,q}\in B$ for $(t,q) \in P$. Thus, when evaluated in~$w$ at a pointed position~$i$ carrying a~``$b$'', we want $\cfinal{\psi}$ to check that $\beta(b)\beta(f(w,i))=t$ and $\delta(b)\delta(f(w,i))=q$. Let $T =\bigl\{(b,t',q')\in B \times M \times N\mid \beta(b)t' = t \text{ and }\delta(b)q' = q \bigr\}$. Using the formulas $\Gamma_{t',q'}$ from Lemma~\ref{lem:checkletters}, we define $\psi = \bigvee_{(b,t',q') \in T} \left(b \wedge \Gamma_{t',q'}\right)$. Boolean operators are handled as expected. It remains to deal with temporal modalities, \emph{i.e.}, the case where there exists $H \in \cdelta$ such that $\psi =  \finallyp{H}{\psi'}$ or $\psi =  \finallymp{H}{\psi'}$. For every $b \in B$, let $F_b = \bigl\{r \in N \mid \delta(b)r \in \delta(H)\bigr\}$. We define:
\[
  \cfinal{\finallyp{H}{\psi'}} \stackrel{\text{def}}{=}
  \left\{\begin{array}{l@{}ll}
      \xfinallyp{B^*}{\bigl(\pi \wedge{}&\left( \bigvee_{b \in B}\left(b \wedge \finallyp{\delta\inv(F_b)}{(\pi \wedge \cfinal{\psi'})}\right)\right)\bigr)} & \text{if $\veps \not\in H$}, \\[1ex]
\xfinallyp{B^*}{\bigl(\pi \wedge{}&\left(\bigvee_{b \in B} \left(b \wedge \finallyp{\delta\inv(F_b)}{(\pi \wedge \cfinal{\psi'})}\right)\vee\cfinal{\psi'}\right)\bigr)} & \text{if $\veps \in H$}.
    \end{array}
  \right.
\]
\[
  \cfinal{\finallymp{H}{\psi'}} \stackrel{\text{def}}{=}
  \left\{
    \begin{array}{l@{\;}ll}
      &\bigvee_{b \in B}  \finallymp{\delta\inv(F_b)}{\left(\pi \wedge b \wedge \xfinallymp{B^*}{\left(\pi \wedge \cfinal{\psi'}\right)}\right)} & \text{if $\veps \not\in H$}, \\
      \xfinallymp{B^*}{\left(\pi \wedge \cfinal{\psi'}\right)} \vee& \bigvee_{b \in B}  \finallymp{\delta\inv(F_b)}{\left(\pi \wedge b \wedge\xfinallymp{B^*}{
      \left(\pi \wedge \cfinal{\psi'}\right)}\right)}  & \text{if $\veps \in H$}.
    \end{array}
  \right.
\]
We give an intuition when $\psi = \finallyp{H}{\psi'}$ and $\veps\notin H$. Let $w_0b_1w_1 \cdots b_nw_n$ be the pointed decomposition of $w$ and $\wch{w}=b'_1\cdots b'_{n}$. Let $i_k\in\pos{w}$ be the position of the \mbox{distinguished~$b_k$}, so that $\mu(i_k)=k$. Now, $\wch{w},k\models\finallyp{H}{\psi'}$ when there exists $m>k$ such that $\wch{w},m\models\psi'$ and $b'_{k+1}\cdots b'_{m-1}\in H$. The construction ensures that $w,i_k\models\cfinal{\finallyp{H}{\psi'}}$ when there exists $m>k$ such that $w,i_m\models\cfinal{\psi'}$ and $b_{k+1}w_{k+1}\cdots b_{m-1}w_{m-1}\in H$. The purpose of using $\xfinallyp{B^*}(\pi\wedge\dots)$ is to ``jump'' to $b_{k+1}$. The remainder checks that the next jump, to a pointed position, determines a word of $\delta\inv(\delta(H))=H$. More generally, one can check that $w,i \models \cfinal{\psi} \Leftrightarrow \wch{w},\mu(i) \models \psi$ for all $w \in K_\ell$ and all pointed positions $i \in \pos{w}$. This concludes the proof.

\section{Natural restrictions of generalized unary temporal logic}
\label{sec:fpcar}
We turn to two natural restrictions of the classes \tlc{\Cs}, which were defined in~\cite{pzfpast}: the pure-future and pure-past fragments.
For a class \Cs, we write $\tlfoa{\Cs} \subseteq \tla{\Cs}$ for the set of all formulas that contain only \emph{future} modalities (\emph{i.e.}, the modalities $\textup{P}_L$ are~disallowed). Symmetrically, $\tlpoa{\Cs} \subseteq \tla{\Cs}$ is the set of all formulas in \tla{\Cs} that contain only \emph{past} modalities (\emph{i.e.}, the modalities $\textup{F}_L$ are~disallowed).

We now define the two associated operators $\Cs \mapsto \tlfoc{\Cs}$ and $\Cs\mapsto\tlpoc{\Cs}$. For every class~\Cs, let \tlfoc{\Cs} be the class consisting of all languages $L_{min}(\varphi)$ where $\varphi \in \tlfoa{\Cs}$. Symmetrically, we write \tlpoc{\Cs} for the class consisting of all languages $L_{max}(\varphi)$, with $\varphi \in \tlpoa{\Cs}$.

\begin{remark} \label{rem:evalwhere}
  Note that \tlfoa{\Cs} formulas are evaluated at the leftmost unlabeled position whereas \tlpoa{\Cs} formulas are evaluated at the rightmost unlabeled position.
\end{remark}

\subsection{Connection with left and right polynomial closure}\label{sec:conn-rpol}

The main ideas to establish decidable characterizations for \tlfoc{\Cs} and \tlpoc{\Cs} follow the lines of the proof of Theorem~\ref{thm:utlmain}. However, there are some differences. First, for the easy direction (proving that some property on \Cs-orbits is necessary), we have to adapt Lemma~\ref{lem:operequ} to the operators $\Cs\mapsto\tlfoc{\Cs}$ and $\Cs\mapsto\tlpoc{\Cs}$. We prove these adapted properties in appendix as corollaries of results presented in~\cite{pzfpast}.

The proof of the difficult direction is mostly identical to that in Theorem~\ref{thm:utlmain}. However, there is a key difference: we have to find a substitute for Proposition~\ref{prop:bcase}, whose proof relied the inclusion $\upol{\bpol{\Cs}}\subseteq\tlc{\Cs}$ from Proposition~\ref{prop:ubptl}. We replace unambiguous polynomial closure (\upolo) by two variants, called \emph{right} and  \emph{left} polynomial closure (\rdeto and \ldeto). It is shown~\cite{pzfpast} that $\rdet{\bpol{\Cs}}\subseteq \tlfoc{\Cs}$ and $\ldet{\bpol{\Cs}}\subseteq \tlpoc{\Cs}$ for every \vari \Cs: this serves as a substitute for Proposition~\ref{prop:ubptl}. Finally, while no simple generic characterization of the classes \rdet{\bpol{\Cs}} and \ldet{\bpol{\Cs}} are known, we are able to replace Theorem~\ref{thm:ubpcar} by combining independent characterizations of the operators \poln and \rdeto (resp.\ \poln and \ldeto) from~\cite{PZ:generic18,pmixed}.

We now establish a connection between the operators $\Cs \mapsto \tlfoc{\Cs}$ and $\Cs \mapsto \tlpoc{\Cs}$ and the two weaker variants \rdeto and \ldeto of unambiguous polynomial closure. Consider a marked product $L_0a_1L_1\cdots a_nL_n$. For $1 \leq i \leq n$, we write $H_i = L_1 a_1L_2 \cdots a_{i-1}L_{i-1}$ and $K_i = L_i a_{i+1} L_{i+1} \cdots a_nL_n$. We say that $L_0a_1L_1\cdots a_nL_n$ is \emph{right deterministic} (resp.\ \emph{left deterministic}) when we have $A^*a_iK_i\cap K_i = \emptyset$ (resp.\ $H_ia_iA^*\cap H_i=\emptyset$) for every $i \leq n$. The \emph{right polynomial closure} of a class \Cs, written \rdet{\Cs}, consists of all \emph{finite disjoint unions} of \emph{right deterministic marked products} $L_0a_1L_1 \cdots a_nL_n$ such that $L_0, \dots,L_n \in \Cs$ (by ``disjoint'' we mean that the languages in the union must be pairwise disjoint). Similarly, the \emph{left polynomial closure} \ldet{\Cs} of \Cs consists of all finite disjoint unions of \emph{left}  deterministic marked products $L_0a_1L_1 \cdots a_nL_n$ such that $L_0, \dots,L_n \in \Cs$. While this is not immediate, it is known~\cite{pmixed} that when the input class \Cs is a \vari, then so are \rdet{\Cs} and \ldet{\Cs}.

As expected, we are interested in the ``combined'' operators $\Cs \mapsto \rdet{\bpol{\Cs}}$ and $\Cs \mapsto \ldet{\bpol{\Cs}}$. Indeed, the first one is connected to the classes \tlfoc{\Cs} by the following result proved in~\cite[Proposition~5]{pzfpast}.

\begin{restatable}{proposition}{lptl} \label{prop:lptl}
  For every \vari \Cs, we have $\rdet{\bpol{\Cs}} \subseteq \tlfoc{\Cs}$.
\end{restatable}

We have the following symmetrical statement for \tlpoc\Cs.
\begin{restatable}{proposition}{rptl}\label{prop:rptl}
  For every \vari \Cs, we have $\ldet{\bpol{\Cs}} \subseteq \tlpoc{\Cs}$.
\end{restatable}

Propositions~\ref{prop:lptl} and ~\ref{prop:rptl} serve as the replacement of Proposition~\ref{prop:ubptl} when dealing with the classes \tlfoc{\Cs} and  \tlpoc{\Cs}, respectively. It now remains to replace the generic algebraic characterization of the classes \upol{\bpol{\Cs}} presented in Theorem~\ref{thm:ubpcar}. This is more tricky as no such characterization is known for the classes \rdet{\bpol{\Cs}} (nor for the classes \ldet{\bpol{\Cs}}). Yet, we~manage to prove a \emph{sufficient condition} for a language to belong to \rdet{\bpol{\Cs}} or \ldet{\bpol\Cs} by combining results of~\cite{pzupol2} and~\cite{pmixed}. While it does \emph{not} characterize these classes in general, it suffices for our needs: proving that particular languages belong to \rdet{\bpol{\Cs}} (and therefore to \tlfoc{\Cs} by Proposition~\ref{prop:lptl}) or to \ldet{\bpol\Cs} (and therefore to \tlfoc{\Cs} by Proposition~\ref{prop:rptl}).

\begin{restatable}{proposition}{rbpcar} \label{prop:rbpcar}
  Let \Cs be a \vari, $L \subseteq A^*$ be a regular language and $\alpha: A^* \to M$ be its syntactic morphism. Assume that $\alpha$ satisfies the following property:
  \begin{equation}\label{eq:rbp}
    (esete)^{\omega+1} = ete(esete)^{\omega} \quad \text{for every \Cs-pair $(e,s) \in M^2$ and every $t \in M$}.
  \end{equation}
  Then, $L \in \rdet{\bpol{\Cs}}$.
\end{restatable}

\begin{restatable}{proposition}{lbpcar} \label{prop:lbpcar}
  Let \Cs be a \vari, $L \subseteq A^*$ be a regular language and $\alpha: A^* \to M$ be its syntactic morphism. Assume that $\alpha$ satisfies the following property:
  \begin{equation*}
    (esete)^{\omega+1} = (esete)^{\omega}ese \quad \text{for every \Cs-pair $(e,t) \in M^2$ and every $s \in M$}.
  \end{equation*}
  Then, $L \in \ldet{\bpol{\Cs}}$.
\end{restatable}

Since Propositions~\ref{prop:rbpcar} and~\ref{prop:lbpcar} are symmetrical, we only prove the first one and leave the second to the reader.

\begin{proof}[Proof of Proposition~\ref{prop:rbpcar}]
  We use a generic characterization of the classes \rdet{\Ds} proved in~\cite{pmixed}. Let us first present it. For every class \Ds, we define a preorder $\preceq_\Ds$ and an equivalence $\sim_\Ds$ over $M$. Given $s,t \in M$, we let,
  \[
    \begin{array}{lcl}
      s \sim_\Ds t & \quad \text{if and only if} \quad  & \text{$s \in F \Leftrightarrow t \in F$ for every $F \subseteq M$ such that $\alpha\inv(F) \in \Ds$,}\\
      s \preceq_\Ds t & \quad \text{if and only if} \quad  & \text{$s \in F \Rightarrow t \in F$ for every $F \subseteq M$ such that $\alpha\inv(F) \in \Ds$.}
    \end{array}
  \]
  Clearly, $\preceq_\Ds$ is a preorder on $M$ and $\sim_\Ds$ is the equivalence generated by $\preceq_\Ds$. When $\alpha: A^* \to M$ is the syntactic morphism of $L$, it is shown in~\cite[Theorem~4.1]{pmixed} that for every \vari \Ds, we have $L \in \rdet{\Ds}$ if and only if $s^{\omega+1} = ts^{\omega}$ for all $s,t \in M$ such that~$s \sim_\Ds t$.

  Hence, since \bpol{\Cs} is a \vari, it suffices to prove that for every $s,t \in M$ such that $s \sim_{\bpol{\Cs}} t$, we have $s^{\omega+1} = ts^{\omega}$. We fix $s,t$ for the proof. Since $s \sim_{\bpol{\Cs}} t$, we have $s \preceq_{\bpol{\Cs}} t$. Moreover, let $\copol{\Cs}$ be the class consisting of all complements of languages in \pol{\Cs} (\emph{i.e.}, $L \in \copol{\Cs}$ if and only if $A^* \setminus L \in \copol{\Cs}$). Clearly, we have $\copol{\Cs} \subseteq \bpol{\Cs}$. Hence, the definition implies that $s \preceq_{\copol{\Cs}} t$

  Moreover, it is shown in~\cite[Lemma~6.6]{pzupol2} that $\preceq_{\copol{\Cs}}$ is the least preorder on $M$ such that for every $x,y,q \in M$ and $e \in E(M)$, if $(e,q) \in M^2$ is a \Cs-pair, then $xeqey \preceq_{\copol{\Cs}} xey$ (the proof is based on the algebraic characterization of \pol{\Cs}, see~\cite{PZ:generic18}). This yields $s_0,\dots,s_n \in M$ such that $s = s_0$, $t = s_n$ and, for every $i \leq n$, there exist $x,y,q \in M$ and $e \in E(M)$ such that  $(e,q) \in M^2$  is a \Cs-pair, $s_{i-1} = xeqey$ and $s_i= xey$. We use induction on $i$ to prove that $s^{\omega+1} = s_is^\omega$ for every $i \leq n$. Since $s_n = t$, the case $i = n$ yields the desired result. When $i = 0$, it is immediate that $s^{\omega+1} = s_0s^\omega$ since $s_0 = s$.  Assume now that $i \geq 1$.

  By induction hypothesis, we know that $s^{\omega+1} = s_{i-1} s^\omega$. Moreover, we have $x,y,q \in M$ and $e \in E(M)$ such that  $(e,q) \in M^2$  is a \Cs-pair, $s_{i-1} = xeqey$ and $s_i= xey$. Since $(s^{\omega+1})^{\omega+2} = s^{\omega+2}$, we get $s^{\omega+2} = (xeqeys^\omega)^{\omega+2}$. Hence, we get
  \[
    \begin{array}{lcl}
      s^{\omega+2}  & = & x\ (eqeys^\omega x e)^{\omega+1}\  eqeys^\omega \\
                    & = & x\ eys^\omega x e(eqeys^\omega x e)^{\omega}\  eqeys^\omega \quad \text{by~\eqref{eq:rbp} since $(e,q)$ is a \Cs-pair} \\
                    & = & xeys^{\omega} (xeqeys^\omega)^{\omega+1}.
    \end{array}
  \]
  This yields, $s^{\omega+2}  =  s_is^{\omega} (s_{i-1}s^\omega)^{\omega+1} = s_is^{\omega} (s^{\omega+1})^{\omega+1}  = s_is^{\omega+1}$. It now remains to multiply by $s^{\omega-1}$ on the right to get $s^{\omega+1} = s_{i} s^\omega$, as desired.
\end{proof}

\subsection{Statements}

The classes \tlfoc{\Cs} and \tlpoc{\Cs} admit algebraic characterizations similar to that of~\tlc{\Cs}. We reuse the \Cs-orbits introduced in Section~\ref{sec:orbits}. Let $\Xs\in \{\Lrel,\Rrel,\Jrel\}$ be one the Green relations defined in Section~\ref{sec:prelims}. A monoid $M$ is \emph{\Xs-trivial} when $s \mathrel{\Xs} t$ implies $s = t$ for all $s,t \in M$. It is standard and simple to verify that a finite monoid $M$ is \Rs-trivial (resp.~ \Ls-trivial) if and only if for all $s,t\in M$, we have $(st)^\omega s=(st)^\omega$ (resp.~ $t(st)^\omega=(st)^\omega$), see~\cite{pinvars,pingoodref} for a proof. We are now able to present the two symmetrical characterizations of \tlfoc{\Cs} and~\tlpoc{\Cs}.

\begin{restatable}{theorem}{fmain} \label{thm:fmain}
  Let \Cs be a \vari, $L \subseteq A^*$ be a regular language and $\alpha: A^* \to M$ be its syntactic morphism. The two following properties are equivalent:
  \begin{enumerate}
    \item $L \in \tlfoc{\Cs}$.
    \item Every \Cs-orbit for $\alpha$ is \Lrel-trivial.
  \end{enumerate}
\end{restatable}

\begin{restatable}{theorem}{pastmain} \label{thm:pastmain}
  Let \Cs be a \vari, $L \subseteq A^*$ be a regular language and $\alpha: A^* \to M$ be its syntactic morphism. The two following properties are equivalent:
  \begin{enumerate}
    \item $L \in \tlpoc{\Cs}$.
    \item Every \Cs-orbit for $\alpha$ is \Rrel-trivial.

  \end{enumerate}

\end{restatable}

Since \tlfoc{\Cs} and \tlpoc{\Cs} are symmetrical, it is natural to consider a third class denoted $\tlfoc{\Cs} \cap \tlpoc{\Cs}$. It consists of all languages belonging simultaneously to \tlfoc{\Cs} and \tlpoc{\Cs}. It is standard that the finite monoids which are both \Lrel-trivial and \Rrel-trivial are exactly the \Jrel-trivial monoids (see~\cite{pinvars,pingoodref}). This yields the following corollary of Theorems~\ref{thm:fmain} and~\ref{thm:pastmain}.

\begin{restatable}{corollary}{intmain} \label{cor:intmain}
  Let \Cs be a \vari, $L \subseteq A^*$ be a regular language and $\alpha: A^* \to M$ be its syntactic morphism. The two following properties are equivalent:
  \begin{enumerate}
    \item $L \in \tlfoc{\Cs} \cap \tlpoc{\Cs}$.
    \item Every \Cs-orbit for $\alpha$ is \Jrel-trivial.
  \end{enumerate}

\end{restatable}

Recall that given a regular language $L \subseteq A^*$ as input, its syntactic morphism $\alpha: A^* \to M$ can be computed. Moreover, Lemma~\ref{lem:orbitcomp} implies that all \Cs-orbits for $\alpha$ can be computed when \mbox{\Cs-separation} is decidable. Thus, the three above characterizations yield the following~corollary.

\begin{restatable}{corollary}{fpmain} \label{cor:fpmain}
  Let \Cs be a \vari with decidable separation. Then, the classes \tlfoc{\Cs}, \tlpoc{\Cs} and $\tlfoc{\Cs} \cap \tlpoc{\Cs}$ have decidable membership.
\end{restatable}

We prove Theorem~\ref{thm:fmain} in the appendix (on the other hand, we omit the proof of Theorem~\ref{thm:pastmain}, which is~symmetrical).

\section{Conclusion}
\label{sec:conc}
We presented generic characterizations of the classes \tlc{\Cs}, \tlfoc{\Cs} and \tlpoc{\Cs}. While~the proofs are complex, the statements are simple and elegant. They generalize in a natural way all known characterizations of classes built with these operators. As a corollary, we obtained that if \Cs is a \vari with decidable \emph{separation}, then all classes \tlc{\Cs}, \tlfoc{\Cs} and \tlpoc{\Cs} \mbox{have decidable \emph{membership}.}

The next step is to tackle \emph{separation}. This question is difficult in general, but it is worth looking at \emph{particular} input classes. For instance, one can define the \tls-hierarchy of basis~\Cs: level $0$ is $\tlnc{0}{\Cs}=\Cs$ and level $n\geq1$ is $\tlnc{n}{\Cs}=\tlc{\tlnc{n-1}{\Cs}}$. It can be shown that the hierarchies of bases $\stzer = \{\emptyset,A^*\}$ and $\dotzer=\{\emptyset,\{\veps\},A^+,A^*\}$ are strict. Thus, since $\bpol\Cs\subseteq\tlc\Cs$, they both classify the star-free languages (or equivalently the languages definable in full linear temporal logic). We already know that in both hierarchies, membership is decidable for levels 1 (\emph{i.e.}, the variants \tls and \tlxs of unary temporal logic) and 2 (which were studied in~\cite{between}). The results of the present paper show that if \tlnc{2}{\stzer} and \tlnc{2}{\dotzer} have decidable separation, then \tlnc{3}{\stzer} and \tlnc{3}{\dotzer} would have decidable~membership.

Finally, all other major operators have language-theoretic counterparts. Another possible follow-up is to look for such a definition for all three operators $\Cs\mapsto\tlc{\Cs}, \tlfoc{\Cs}$ and $\tlpoc{\Cs}$.

\bibliography{main}

\newpage
\appendix

\section{Appendix to Section~\ref{sec:fpcar}}
\label{app:fpcar}
This appendix is devoted to the proof of Theorem~\ref{thm:fmain} (on the other hand, we do not prove Theorem~\ref{thm:pastmain}, since it is symmetrical). The proof involves some preliminary work because we have to adapt the results that we presented for the classes \tlc{\Cs} in Section~\ref{sec:utl} to the classes $\tlfoc{\Cs}$. More precisely, we first prove a characteristic property of the classes \tlfoc{\Cs} in Section~\ref{sec:char-FL}. Section~\ref{sec:conn-rpol} is devoted to results that serve as replacements for Theorems~\ref{prop:ubptl} and~\ref{thm:ubpcar}. Finally, the characterization of \tlfoc\Cs is presented in Section~\ref{sec:proof-theo-fl}.

\subsection{Characteristic property of the classes \tlfoc{\Cs}}\label{sec:char-FL}

We associate canonical equivalences to the classes \tlfoc{\Cs} where \Cs is a \vari. The definition is taken from~\cite{pzfpast} and is similar to those defined for the classes \tlc{\Cs}. In particular, recall that given a morphism $\eta: A^*\to N$ into a finite monoid $N$, we write $\Cs_\eta$ for the class of all languages recognized by $\eta$. We also reuse the notion of rank that we defined for temporal formulas in Section~\ref{sec:utlcar}.

Consider a morphism $\eta: A^* \to N$ into a finite monoid and $k \in \nat$. Let $\eta: A^* \to N$  Given $w,w' \in A^*$, $i \in \pos{w}$ and $i'\in \pos{w'}$, we write, $w,i \tlfeqk w',i'$ when:
\[
  \text{For every \tlfoa{\Cs_\eta} formula $\varphi$ with rank at most $k$,} \quad w,i \models \varphi \Longleftrightarrow w',i' \models \varphi.
\]
The relations \tlfeqk are equivalences. It is also immediate from the definition and Lemma~\ref{lem:rank}, that they have finite index. Finally, we also introduce equivalences which compare single words of $A^*$. Abusing terminology, we also write them \tlfeqk. Given $w,w' \in A^*$, we write $w \tlfeqk w'$ if $w,0 \tlfeqk w',0$. Clearly, the relation $\tlfeqk$ is an equivalence relation over $A^{*}$. We use it to characterize the classes \tlfoc{\Cs} when the class \Cs is a \emph{\vari}. More precisely, we have the following lemma whose proof is identical to that of Lemma~\ref{lem:operequ}.

\begin{restatable}{lemma}{operequf}\label{lem:operequf}
  Let \Cs be a \vari and $L \subseteq A^*$. Then, $L \in \tlfoc{\Cs}$ if and only if there exists a \Cs-morphism $\eta: A^* \to N$ and $k \in \nat$ such that $L$ is a union of \tlfeqk-classes.
\end{restatable}

We now present a few key lemmas that we shall use to prove the characteristic property of the classes \tlfoc{\Cs}. The first one states that the equivalences \tlfeqk are congruences. This can be verified from the definition.

\begin{restatable}{lemma}{tlfcong} \label{lem:tlfcong}
  Consider a morphism $\eta: A^* \to N$ into a finite monoid and $k \in \nat$. For every $u,v,u',v' \in A^*$ such that $u \tlfeqk u'$ and $v \tlfeqk v'$, we have $uv \tlfeqk u'v'$.
\end{restatable}

Moreover, we have the following useful lemma, which involves auxiliary alphabets. It can also be verified from the definitions.

\begin{restatable}{lemma}{fpmorph} \label{lem:fpmorph}
  Let $A,B$ be two alphabets and let $\gamma: B^* \to A^*$ be a morphism. Let $\eta: A^* \to N$ be a morphism into a finite monoid and let $\delta= \eta \circ \gamma: B^* \to N$. For all $w,w' \in B^*$ and $k \in \nat$ such that $w \tlfdqk w'$, we have $\gamma(w) \tlfeqk \gamma(w')$.
\end{restatable}

Finally, we adapt Proposition~\ref{prop:efg} to the equivalences \tlfeqk.

\begin{restatable}{proposition}{fefg} \label{prop:fefg}
  Consider a morphism $\eta: A^* \to N$ into a finite monoid, let $e \in E(N)$ be an idempotent and let $u,v,z \in \eta\inv(e)$. For every $k \in \nat$, the following property holds:
  \[
    (z^{k}uz^{2k}vz^{k})^{k} \tlfeqk z^kvz^k(z^{k}uz^{2k}vz^{k})^{k}.
  \]
\end{restatable}

\begin{proof}
  Note first that proving the property boils down to the special case when $u,v$ and $z$ are single letter words. Indeed, let $B = \{a,b,c\}$ and consider the morphism $\gamma: B^* \to A^*$ defined by $\gamma(a) = u$, $\gamma(b) = v$ and $\gamma(c) = z$. Moreover, let $\delta= \eta \circ \alpha: B^* \to N$. We prove that $(c^{k}ac^{2k}bc^{k})^{k} \tlfdqk c^kbc^k(c^{k}ac^{2k}bc^{k})^{k}$. It will then follow from Lemma~\ref{lem:fpmorph} that $(z^{k}uz^{2k}vz^{k})^{k} \tlfeqk z^kvz^k(z^{k}uz^{2k}vz^{k})^{k}$, which will conclude the proof.

  We use induction to prove a stronger property. Let us start with a preliminary definition. Let $w,w' \in B^*$, $i \in \pos{w}$, $i' \in \pos{w'}$ and $\ell \in\nat$ we write $w,i \sim_\ell w',i'$ if and only if one the two following conditions is satisfied.
  \begin{enumerate}
    \item $|w|+1 - i = |w'|+1 - i'$, $\wpos{w}{i} = \wpos{w'}{i'}$ and, if $i < |w|+1$, $\suffix{w}{i} = \suffix{w'}{i'}$.
    \item $i < |w|+1 -\ell$, $i' < |w'|+1 - \ell$, $\wpos{w}{i} = \wpos{w'}{i'}$ and $\infix{w}{i}{i+\ell} = \infix{w'}{i'}{i'+\ell}$.
  \end{enumerate}
  We can now present the general property that we shall prove by induction. We write $x = c^kac^k$ and $y = c^kbc^k$ for the proof. Let $\ell \in \nat$ and consider a quadruple $(w,i,w',i')$ where $w,w' \in A^*$, $i \in \pos{w}$ and $i' \in \pos{w'}$. We say that $(w,i,w',i')$ is an \emph{$\ell$-candidate} if there exist $m\geq\ell$ and $w_1,w'_1 \in (x+y)^*$ such that $w=w_1(xy)^m $, $w'=w'_1(xy)^m$, $i \leq |w_1|$, $i' \leq |w'_1|$ and $w,i \sim_\ell w',i'$.

  \begin{restatable}{lemma}{lemfefg} \label{lem:fefg}
    For every $\ell \leq k$, every $\ell$-candidate $(w,i,w',i')$ and every formula $\varphi \in \tlfoa{\Cs_\delta}$ of rank at most $\ell$, we have $w,i\models\varphi \Leftrightarrow w',i' \models \varphi$.
  \end{restatable}

  Let us first apply Lemma~\ref{lem:fefg} to complete the main proof. Clearly, $((xy)^k,0,y(xy)^k,0)$ is a $k$-candidate (recall that both $x$ and $y$ start with the prefix~$c^k$). Hence, the lemma implies that for all $\varphi \in \tlfoa{\Cs_\delta}$ of rank at most $k$, we have $(xy)^k,0\models\varphi \Leftrightarrow y(xy)^k,0 \models \varphi$. This exactly says that $(c^{k}ac^{2k}bc^{k})^{k} \tlfdqk c^kbc^k(c^{k}ac^{2k}bc^{k})^{k}$ which completes the main proof.

  \smallskip

  We now concentrate on the proof of Lemma~\ref{lem:fefg}. We fix $\ell \leq k$, an $\ell$-candidate $(w,i,w',i')$ and a formula $\varphi \in \tlfoa{\Cs_\delta}$ of rank at most $\ell$ for the proof. We use induction on the size of~$\varphi$ to prove that $w,i\models\varphi\Leftrightarrow w',i'\models\varphi$. Assume first that $\varphi$ is an atomic formula. Since $(w,i,w',i')$ is an $\ell$-candidate, we have $w,i \sim_\ell w',i'$, which yields $\wpos{w}{i} = \wpos{w'}{i'}$. Hence, $w,i\models\varphi\Leftrightarrow w',i'\models\varphi$ since $\varphi$ is atomic. Boolean connectives are handled in the natural way using induction on the size of $\varphi$. It remains to handle the case when $\varphi = \finallyl{\psi}$ for $L \in \Cs_\delta$ and $\psi \in \tlfoa{\Cs_\delta}$ has rank at most $\ell-1$. By symmetry, we only prove that $w,i\models\varphi \Rightarrow w',i'\models\varphi$. Hence, we assume that $w,i \models \varphi$. By hypothesis on $\varphi$, this yields $j \in \pos{w}$ such that $i < j$, $\infix{w}{i}{j} \in L$ and $w,j \models \psi$. We use $j \in \pos{w}$ to construct a position $j' \in \pos{w'}$ such that $i' < j'$, $\delta(\infix{w}{i'}{j'}) = \delta(\infix{w}{i}{j})$ (which yields $\infix{w}{i'}{j'} \in L$ since $L$ is recognized by $\delta$) and $w',j' \models \psi$. This will imply that $w',i'\models\varphi$, as desired. Since $(w,i,w',i')$ is an $\ell$-candidate, we have $m \geq\ell$ and $w_1,w'_1 \in (x+y)^*$ such that $w=w_1(xy)^m $, $w'=w'_1(xy)^m$, $i \leq |w_1|$, $i' \leq |w'_1|$ and $w,i \sim_\ell w',i'$. We consider three cases depending on the position $j \in \pos{w}$.

  \smallskip
  \noindent
  \textbf{First case: $i+1 = j$.} We define $j' = i' + 1$. Clearly, $i' < j'$ and $\infix{w'}{j'}{i'}=\infix{w}{j}{i} = \veps$. Hence, we have to verify that $w',j' \models \psi$. Since $\psi$ has rank at most $\ell-1$ and $w,j \models \psi$, it suffices to prove that $(w,j,w',j')$ is an $(\ell-1)$-candidate. It will then follow from induction that $w',j' \models \psi$. Since $m \geq \ell$, we have $m-1 \geq \ell-1$. Moreover, we have $w = w_1xy(xy)^{m-1}$ and $w' = w'_1xy(xy)^{m-1}$. Since $j' = i+1$, $j' = i'+1$, $i \leq |w_1|$ and $i' \leq |w'_1|$, we have $j \leq |w_1xy|$  and $j' \leq |w'_1xy|$. Finally, since $w,i \sim_\ell w',i'$,  $j = i+1$ and $j' = i' +1$, one can verify that $w,j \sim_{\ell-1} w',j'$. Thus,  $(w,j,w',j')$ is indeed an $(\ell-1)$-candidate.

  \smallskip
  \noindent
  \textbf{Second case: $i+1 < j$ and $j \leq |w_1xy|$.} We know that $i' \leq |w'_1|$. Hence, by definition of $x,y$ and since $\ell \leq k$, one can check that there exists $j' \in \pos{w'}$ such that $j' \leq |w'_1xy|$, $i'+1 < j'$ and $w,j \sim_{\ell-1} w',j'$. We have $i' < j'$ by definition. Moreover, since $i+1 < j$ and $i'+1 < j'$, we know that $\infix{w'}{i'}{j'}$ and $\infix{w}{i}{j}$ are nonempty, which yields $\delta(\infix{w'}{i'}{j'}) = \delta(\infix{w}{i}{j})$ by definition of $\delta$. Finally, it is immediate that $(w,j,w',j')$ is an $(\ell-1)$-candidate. Indeed, we have $w = w_1xy(xy)^{m-1}$, $w' = w'_1xy(xy)^{m-1}$,  $j \leq |w_1xy|$, $j' \leq |w'_1xy|$ and $w,j \sim_{\ell-1} w',j'$. Thus, since $w,j \models \psi$, induction yields $w',j' \models \psi$ which completes this case.

  \smallskip
  \noindent
  \textbf{Third case: $i+1 < j$ and $|w_1xy| < j$.} In this case, $j$ is inside the suffix $(xy)^{m-1}$ of $w = w_1(xy)^m$. We define $j' \in \pos{w'}$ as the corresponding position in the suffix $(xy)^{m-1}$ of $w' = w'_1(xy)^m$ (\emph{i.e.}, $j' = j + |w'| - |w|$). Since $i' \leq |w_1|$, it is clear that $i' < j'$. In fact, we have $i' +1 < j'$ which means that $\infix{w'}{i'}{j'}$ and $\infix{w}{i}{j}$ are nonempty. This yields $\delta(\infix{w'}{j'}{i'}) = \delta(\infix{w}{j}{i})$ by definition of $\delta$. Finally, we have $\suffix{w}{j} = \suffix{w'}{j'}$ by definition of~$j'$. Hence, since $\psi \in \tlfoa{\Cs_\delta}$ (and therefore only contains future modalities), we get $w',j' \models \psi$ since we already know that $w,j \models \psi$. This completes the proof.
\end{proof}

We are ready to prove the characteristic property of the classes \tlfoc{\Cs}. It is a simple corollary of Proposition~\ref{prop:fefg}.

\begin{restatable}{proposition}{fefgame} \label{prop:fefgame}
  Let \Cs be a \vari and $L \in \tlfoc{\Cs}$. There exists a \Cs-morphism  $\eta: A^* \to N$  and $k \in \nat$ such that for every idempotent $f \in E(N)$, every $u,v,z \in \eta\inv(f)$ and every $x,y \in A^*$, the following property holds:
  \[
    x(z^{k}uz^{2k}vz^{k})^{k}y \in L \quad \text{if and only if} \quad x z^kvz^k(z^{k}uz^{2k}vz^{k})^{k}y \in L.
  \]
\end{restatable}

\begin{proof}
  Lemma~\ref{lem:operequf} yields a \Cs-morphism $\eta: A^* \to N$ and $k \in \nat$ such that $L$ is a union of \tlfeqk-classes. Now, consider $u,v,z \in A^*$ such that $\eta(u) = \eta(v) = \eta(z) \in E(N)$ and $x,y \in A^*$. We have to prove that,
  \[
    x(z^{k}uz^{2k}vz^{k})^{k}y \in L \quad \text{if and only if} \quad x z^kvz^k(z^{k}uz^{2k}vz^{k})^{k}y \in L.
  \]
  It is immediate from Lemma~\ref{lem:tlfcong} and Proposition~\ref{prop:fefg} that,
  \[
    x(z^{k}uz^{2k}vz^{k})^{k}y \tlfeqk xz^kvz^k(z^{k}uz^{2k}vz^{k})^{k}y.
  \]
  Since $L$ is a union of \tlfeqk-classes, the desired result follows.
\end{proof}

\subsection{Proof of Theorem~\ref{thm:fmain}}\label{sec:proof-theo-fl}

Let us first recall the statement of Theorem~\ref{thm:fmain}.

\fmain*

\noindent
Fix a \vari~\Cs, a regular language $L \subseteq A^*$ and its syntactic morphism $\alpha: A^* \to M$ for the proof. We start with the left to right implication, which follows directly from Proposition~\ref{prop:fefgame}.

\medskip
\noindent
\textbf{From \tlfoc{\Cs} to \Lrel-triviality.} Assume that $L \in \tlfoc{\Cs}$. For every $e \in E(M)$, we prove that the \Cs-orbit $M_e$ of $e$ for $\alpha$ is \Lrel-trivial. We fix $e$ for the proof. Let $s,t \in M_e$, we have to show that $t(st)^\omega =(st)^\omega$. Since $L \in \tlfoc{\Cs}$, Proposition~\ref{prop:fefgame} yields a \Cs-morphism $\eta: A^* \to N$ and $k \in \nat$ such that for every $f \in E(N)$, every $u,v,z \in \eta\inv(f)$ and every $x,y \in A^*$, we have,
\begin{equation} \label{eq:ftlcnec}
  x(z^{k}uz^{2k}vz^{k})^{k}y \in L \quad \Leftrightarrow\quad xz^kvz^k(z^{k}uz^{2k}vz^{k})^{k}y \in L.
\end{equation}
Since $\eta$ is a \Cs-morphism, Lemma~\ref{lem:orbitnec} yields $f\! \in\! E(N)$ such that $M_e\! \subseteq\! \alpha(\eta\inv(f))$. As \mbox{$e,s,t \in M_e$}, we get $u,v,z \in A^*$ such that $u,v,z \in \eta\inv(f)$, $\alpha(u) = s$, $\alpha(v) = t$ and $\alpha(z) = e$. Thus, the words $u,v,z \in A^*$ satisfy~\eqref{eq:ftlcnec} and the words  $(z^kuz^kvz^k)^{k}$ and $z^kvz^k(z^kuz^kvz^k)^{k}$ are equivalent for the syntactic congruence of $L$. Hence, they have the same image under~$\alpha$. Since $e \in E(M)$, this yields $(esete)^{k} = ete(esete)^{k}$. Hence, $(st)^{k} = t(st)^{k}$ since $e,s,t \in M_e$ and $e$ is neutral in $M_e$ by Lemma~\ref{lem:orbitmono}. It now suffices to multiply by enough copies of $st$ the right to get $(st)^\omega = t (st)^\omega$, which completes the proof.

\medskip
\noindent
\textbf{From \Lrel-triviality to \tlfoc{\Cs}.} Assuming that every \Cs-orbit for $\alpha$ is \Lrel-trivial, we show $L \in \tlfoc{\Cs}$. We use induction to build a \tlfoc{\Cs} formula defining $L$. The proof is similar to that of the corresponding implication in Theorem~\ref{thm:utlmain}. We start with preliminary terminology.

Lemma~\ref{lem:pairmor} yields a \Cs-morphism $\eta: A^* \to N$ such that the \Cs-pairs for $\alpha$ are exactly the $\eta$-pairs for $\alpha$. We fix $\eta$ for the proof. We define an auxiliary alphabet $B$. Let $P \subseteq M \times N$ be the set of all pairs $(\alpha(w),\eta(w)) \in M\times N$ where $w \in A^+$ is a nonempty word. For each pair $(s,r) \in P$, we create a fresh letter $b_{s,r} \not\in A$ and define $B = \{b_{s,r} \mid (s,r) \in P\}$.

Let $\beta: B^* \to M$ and $\delta: B^* \to N$ be the morphisms defined by $\beta(b_{s,r}) = s$ and $\delta(b_{s,r}) = r$ for $(s,r) \in P$. By definition, it remains true that for every $w \in B^+$, we have $(\beta(w),\delta(w)) \in P$. Let \cdelta be the class of all languages (over $B$) recognized by $\delta$. One can check that \cdelta is a \vari. We have the following simple lemma whose proof is identical to that of Lemma~\ref{lem:beta}.

\begin{restatable}{lemma}{thebeta2} \label{lem:beta2}
  For every $F \subseteq M$, if $\beta\inv(F) \in \tlfoc{\cdelta}$, then $\alpha\inv(F) \in \tlfoc{\Cs}$.

\end{restatable}

In view of Lemma~\ref{lem:beta2}, it now suffices to prove that every language recognized by $\beta$ belongs to \tlfoc{\cdelta}. Since $L$ is recognized by $\alpha$, the lemma will then imply that $L \in \tlfoc{\Cs}$, as desired.

We also reformulate our hypothesis that every \Cs-orbit for $\alpha$ is \Lrel-trivial using $\beta$ and $\delta$. The proof is identical to that of Lemma~\ref{lem:hypbeta}.

\begin{restatable}{lemma}{hypbeta2} \label{lem:hypbeta2}
  For every $e \in E(M)$ and every $s,t \in M$, if $(e,s)$ and $(e,t)$ are $\delta$-pairs for $\beta$, then $(esete)^{\omega} = ete(esete)^{\omega}$.
\end{restatable}

We next recall the definition of $\dpj{r} \in \nat$ associated to elements $r\in N$. We let $\dpj{r}$ be the maximal integer $n\in\nat$ such that there exist $n$ elements $r_1,\dots,r_n \in N$ satisfying $r \Jords r_1 \Jords \cdots \Jords r_n$. By definition, $0 \leq \dpj{r} \leq |N| - 1$. In particular, we have $\dpj{r} = 0$ if and only if $r$ is maximal for \Jord (\emph{i.e.}, if and only if $r \Jrel 1_N$). Finally, given a word $w \in B^*$, we write $\dpj{w} \in \nat$ for $\dpj{\delta(w)}$. Recall that for all $x,y,z \in B^*$, we have $\dpj{y} \leq \dpj{xyz}$.
Finally, we use the same family of languages $K_\ell \subseteq B^*$ ($\ell \in \nat$) as in the main text:
\[
  K_\ell = \{w \in B^* \mid \text{for all $k \leq \ell$ and $x,y,z \in B^*$, if $w = xyz$ and $|y| = k$, then $\dpj{y} \geq k$}\}.
\]
Recall that $K_0 = B^*$ and that for every $\ell \geq |N|$, the language $K_{\ell}$ is finite.

Our goal is to prove that for all $s \in M$ and $\ell \in\nat$, we have $K_{\ell}\cap\beta\inv(s) \in \tlfoc{\cdelta}$. The case $\ell = 0$ will then imply that every language recognized by $\beta$ indeed belongs to \tlfoc{\cdelta} (the others cases are important for the induction). The proof is very similar to that of Theorem~\ref{thm:utlmain}. It involves two steps. The first is based on the link between \tlfoc{\cdelta} and \rdet{\bpol{\cdelta}}. This is the main difference with the corresponding part in the proof of Theorem~\ref{thm:utlmain} (the remainder of the proof is identical).

\begin{restatable}{proposition}{fbcase} \label{prop:fbcase}
  Let  $\ell\! \in\! \nat$, $s\! \in\! M$ and $r\! \in\! N$. If $\dpj{r}\! \leq\! \ell$, then $K_{\ell} \cap \beta\inv(s) \cap \delta\inv(r) \!\in\! \tlfoc{\cdelta}$.
\end{restatable}

\begin{proof}
  We prove that $K_{\ell} \cap \beta\inv(s) \cap \delta\inv(r) \in \rdet{\bpol{\cdelta}}$. By Proposition~\ref{prop:lptl}, this yields $K_{\ell} \cap \beta\inv(s) \cap \delta\inv(r)\in\tlfoc{\cdelta}$ as desired.  Let $\gamma: B^*\to Q$ be the syntactic morphism of $K_{\ell} \cap \beta\inv(s)\cap\delta\inv(r)$. By Proposition~\ref{prop:rbpcar}, it suffices to show that given $q_1,q_2\in Q$ and $f \in E(Q)$ such that $(f,q_1) \in Q^2$ is a \cdelta-pair for $\gamma$, the following equation holds:
  \begin{equation} \label{eq:fbasec}
    (fq_1fq_2f)^{\omega+1} = fq_2f (fq_1fq_2f)^{\omega}.
  \end{equation}
  By definition of \cdelta, we know that $\delta$ is a $\cdelta$-morphism. Therefore, Lemma~\ref{lem:pairmor} implies that $(f,q_1)$ is a $\delta$-pair and we get $u',v'_1 \in B^*$, such that $\delta(u') = \delta(v'_1)$, $\gamma(u')=f$ and $\gamma(v'_1)=q_1$. Observe that if $v'_1 = \veps$, then $q_1 = 1_Q$ and~\eqref{eq:fbasec} holds since both sides are equal to $(fq_2f)^{\omega+1}$. Hence, we assume from now on that $v'_1 \in B^+$. Also, we fix $v'_2 \in B^*$ such that $\gamma(v'_2) = q_2$. We now define $p = \ell \times \omega(N) \times \omega(M) \times \omega(Q)$, $u = (u')^p$, $v_1 = (u')^{p-1}v'_1$ and $v_2 = uv'_2u(uv_1uv'_2u)^{p-1}$. We compute $\gamma(u) = f$, $\gamma(v_1) = fq_1$ and $\delta(u) = \delta(v_1)$. Moreover, since $p$ is a multiple of $\omega(N)$, the element $\delta(u) = \delta(v_1)$ is an idempotent $g \in E(N)$. Finally, we have $\gamma(v_2) = fq_2f(fq_1fq_2f)^{p-1}$ and $\delta(v_2) = (g\delta(v'_2)g)^p$. In particular, it follows that $\delta(v_2)$ is an idempotent $h \in E(N)$ such that $gh = hg = h$. We prove that $(uv_1uv_2u)^{p}$ and  $uv_2u (uv_1uv_2u)^{p}$ are equivalent for the syntactic congruence of $K_{\ell} \cap \beta\inv(s)\cap \delta\inv(r)$. This will imply that they have the same image under $\gamma$, which yields $(fq_1fq_2f)^{\omega}  = fq_2f(fq_1fq_2f)^{2\omega-1}$, and the desired equation~\eqref{eq:fbasec} follows by multiplying on the right by $fq_1fq_2f$. Let $x,y \in A^*$, $z_1 = x (uv_1uv_2u)^{p} y$ and $z_2 = xuv_2u (uv_1uv_2u)^{p}y$. We have to show that $z_1 \in K_{\ell} \cap \beta\inv(s)\cap\delta\inv(r)$ if and only if $z_2 \in K_{\ell}  \cap \beta\inv(s)\cap\delta\inv(r)$. We first treat a special case.

  Assume that $|u| < \ell$ and $\ell \geq 1$. We show that in this case $z_1 \not\in K_\ell$ and $z_2 \not\in K_\ell$ (which implies the desired result). Since $u = (u')^p$ and $p\geq\ell$, the hypothesis that $|u| < \ell$ yields $u = u' = \veps$. Since $\delta(u)= \delta(v_1)$, we get $\delta(v_1) = 1_N$. Recall that $v_1 = (u')^{p-1}v'_1$ and $v'_1 \in B^+$ by hypothesis. Therefore, $v_1\in B^+$, which means that $v_1$ contains a letter $b \in B$ such that $\delta(b) \Jrel 1_N$. In particular $\dpj{b} = 0$. Hence, $b$ is an infix of length $1$ of both $z_1$ and $z_2$ such that $\dpj{b} < 1$. Since $\ell \geq 1$, this yields $z_1 \not\in K_\ell$ and $z_2 \not\in K_\ell$.

  This completes the special case. We assume from now on that either $|u| \geq \ell$ or $\ell = 0$. Since $\delta(u) = \delta(v_1) = g \in E(N)$, and $\delta(v_2)=h\in E(N)$ with $gh=hg=h$, we get $\delta(z_1) = \delta(z_2) = \delta(x)h\delta(y)$.  Thus, $z_1 \in  \delta\inv(r)$ if and only if $z_2 \in \delta\inv(r)$. Let us now prove that $z_1 \in K_{\ell} \Leftrightarrow z_2 \in K_{\ell}$. This is trivial if $\ell = 0$ since $K_0 = B^*$. Thus, we assume that $\ell \geq 1$. In this case, we have $|u| \geq \ell$ by hypothesis. It follows that for every $k \leq \ell$, $z_1$ and $z_2$ have the same infixes of length $\ell$. This implies $z_1 \in K_{\ell} \Leftrightarrow z_2 \in K_{\ell}$, as desired.

  It remains to prove that if $z_1,z_2 \in K_\ell \cap \delta\inv(r)$, then $\beta(z_1) = \beta(z_2)$. We first show that these hypotheses yield $g \Jrel h$. There are two cases. First, assume that $\ell = 0$. Since $\dpj{r} \leq \ell$ by hypothesis, we get $r \Jrel 1_N$. Thus, since $u$ and $v_2$ are infixes of $z_1 \in \delta\inv(r)$, we have $\delta(u) \Jrel \delta(v_2) \Jrel 1_N$ which exactly says that $g \Jrel h \Jrel 1_N$. Assume now that $\ell \geq 1$. By hypothesis, this yields $|u| \geq \ell$. Since $u$ is an infix of $v_2$, this also implies that $|v_2| \geq \ell$. Hence, since $u$ and $v_2$ are infixes of $z_2 \in K_\ell \cap \delta\inv(r)$, we get $\dpj{u} \geq \ell$ and $\dpj{v_2} \geq \ell$ (by definition of $K_\ell$), $r \Jord \delta(u)$ and $r \Jord \delta(v_2)$. In particular, it follows that $\dpj{r} \geq \dpj{u} \geq \ell$ and $\dpj{r} \geq \dpj{v_2} \geq \ell$. Since $\dpj{r} \leq \ell$ by hypothesis on $r$, we get $\dpj{r} = \dpj{u} = \dpj{v_2} =\ell$. Together with $r \Jord \delta(u)$ and $r \Jord \delta(v_2)$, this yields $r \Jrel \delta(u) \Jrel \delta(v_2)$, \emph{i.e.}, that $r \Jrel g \Jrel h$. This completes the proof that $g \Jrel h$. Since we also know that $hg = gh = h$, we have $h \Rord g$ and Lemma~\ref{lem:jlr} yields $g \Rrel h$. We get $z \in N$ such that $g = hz$. Therefore, we have $h = hg = hhz = hz = g$.

  Altogether, we obtain $\delta(u) = \delta(v_1) = \delta(v_2) = g \in E(N)$. This implies that $(\beta(u),\beta(v_1))$ and $(\beta(u),\beta(v_2))$ are $\delta$-pairs for $\beta$. Moreover, recall that $u = (u')^p$ where $p$ is a multiple of $\omega(M)$. Therefore, we have $\beta(u) \in E(M)$. Consequently, it follows from Lemma~\ref{lem:hypbeta} that $\beta((uv_1uv_2u)^p) = \beta(uv_2u(uv_1uv_2u)^p)$. It now suffices to multiply by $\beta(x)$ on the left and $\beta(y)$ on the right to obtain $\beta(z_1) = \beta(z_2)$, as desired.
\end{proof}

We now turn to the second step of the proof. It is formalized in the following statement.

\begin{restatable}{proposition}{propfpmain} \label{prop:fpmain}
  Let $\ell \leq |M|$ and $s \in M$. There exists an \tlfoa{\cdelta} formula $\varphi_{\ell,s}$ such that for every $w \in K_\ell$, we have $w,0 \models \varphi_{\ell,s} \Leftrightarrow \beta(w) = s$.
\end{restatable}

We omit the proof of Proposition~\ref{prop:fpmain} as the argument is basically a copy and paste of the one we presented for Proposition~\ref{prop:pmain} in the main text (one has to discard the cases involving past modalities and replace Proposition~\ref{prop:bcase} with Proposition~\ref{prop:fbcase}). Let us explain why Proposition~\ref{prop:fpmain} completes the main proof. We show that every language recognized by $\beta$ belongs to \tlfoc{\cdelta}. Clearly, it suffices to show that $\beta\inv(s) \in \tlfoc{\cdelta}$ for all $s \in M$. We apply Proposition~\ref{prop:fpmain} for $\ell = 0$. Since $K_0=B^*$, this yields a formula $\varphi_{\ell,s} \in \tlfoa{\cdelta}$ such that $L_{min}(\varphi_{\ell,s}) = \beta\inv(s)$. Thus, $\beta\inv(s) \in \tlfoc{\cdelta}$, as desired.

 \end{document}